\documentclass[letterpaper,twocolumn,10pt]{article}
\usepackage{usenix}

\usepackage{pifont}
\newcommand{\cmark}{\textcolor{green!80!black}{\ding{51}}} 
\newcommand{\xmark}{\textcolor{red}{\ding{55}}} 

\usepackage{amsmath}
\usepackage{amssymb}
\usepackage{amsthm}
\usepackage{mathtools}  

\usepackage{marvosym}
\usepackage{comment}
\usepackage{bm}

\usepackage{makecell}
\usepackage{colortbl} 

\usepackage{tikz}
\usetikzlibrary{decorations.pathreplacing, positioning}
\usepackage{graphicx}

\usepackage{booktabs}
\usepackage{multirow}

\usepackage{enumitem}

\usepackage{algorithm}
\usepackage{algpseudocode}

\algnewcommand\algorithmicparfor{\textbf{parallel for}}
\algdef{S}[FOR]{ParFor}[1]{\algorithmicparfor\ #1\ \algorithmicdo}
\algdef{E}[FOR]{EndParFor}{\algorithmicend\ \algorithmicparfor}

\newcommand{\PhaseComment}[1]{\Statex \textcolor{teal}{$\triangleright$ \textit{#1}}}

\usepackage{listings}

\usepackage{hyperref}
\usepackage{url}

\usepackage{cite}

\usepackage{amssymb}

\newcommand{\mhot}{\textsc{Mhot}}
\newcommand{\hot}{\textsc{HOT}}
\newcommand{\mpt}{\textsc{MPT}}
\newcommand{\lvmt}{\textsc{LVMT}}

\newtheorem{theorem}{Theorem}
\newtheorem{lemma}[theorem]{Lemma}
\newtheorem{corollary}[theorem]{Corollary}
\newtheorem{definition}{Definition}

\newtheorem{remark}[theorem]{Remark}

\newenvironment{packeditemize}{
	\begin{list}{$\bullet$}{
			\setlength{\labelwidth}{4pt}
			\setlength{\itemsep}{0pt}
			\setlength{\leftmargin}{\labelwidth}
			\addtolength{\leftmargin}{\labelsep}
			\setlength{\parindent}{0pt}
			\setlength{\listparindent}{\parindent}
			\setlength{\parsep}{0pt}
			\setlength{\topsep}{1pt}}}{\end{list}}

\newcommand{\heading}[1]{{\vspace{3pt}\noindent{\textbf{#1}}}}

\begin{document}

\date{}

\title{\Large \bf \textsc{Mhot}: Height-Optimized Authenticated Data Structure\\for Blockchain State Commitment}

\author{
{\rm Sipeng Xie\thanks{Beihang University}, \quad
Qianhong Wu\footnotemark[1], \quad
Minghang Li\footnotemark[1], \quad
Qiyuan Gao\footnotemark[1], \quad
Bo Qin\thanks{Renmin University of China (\Letter)}, \quad
Qin Wang\thanks{Independent}}
}

\maketitle

{\let\thefootnote\relax\footnotetext{Accepted by \textcolor{violet}{USENIX Security 2026}}}

\begin{abstract}

State root computation dominates (\textasciitilde78\%) blockchain block processing time. Ethereum's canonical authenticated data structure, i.e., Merkle Patricia Trie (MPT), suffers from severe tree-height growth and is vulnerable to \textit{Nurgle attacks} (S\&P'24), where adversaries inflate path depth via hash collisions and degrade system performance at negligible cost. Existing defenses increase node fanout (span) to bound tree height, but higher span inflates proof size exponentially. Prior work mitigates this trade-off using vector commitments, at the cost of trusted setup or expensive verification.

We present \textsc{Mhot}, a height-optimal authenticated data structure for blockchain state commitment that preserves standard hash-based verification without trusted setup. Unlike MPT’s fixed-prefix indexing, which couples span and fanout exponentially, \textsc{Mhot} indexes by \emph{discriminative bits} that actually distinguish keys, achieving adaptive span with linear fanout coupling and provably minimal height. To prevent high fanout from inflating proofs, we introduce \emph{hierarchical proofs}, a two-layer Merkle construction that reduces per-node proof overhead from $O(k)$ to $O(\log k)$. 

On Ethereum mainnet workloads, \textsc{Mhot} achieves up to \textbf{9$\times$} higher write throughput, \textbf{4$\times$} lower write amplification, and \textbf{2$\times$} smaller proofs than MPT. Under Nurgle attacks, even when the adversary consumes an entire block’s gas budget, \textsc{Mhot} maintains a \textbf{0\%} attack success rate (v.s., 99.97\% for MPT). Our results, somewhat surprisingly, show that height optimality (not new crypto primitives!) is the key abstraction for scalable and attack-resilient blockchain state commitment.

\end{abstract}
\section{Introduction}
\label{sec:introduction}

State root computation consumes 70--80\% of total block processing time in modern blockchain systems~\cite{chainkv,lmpts,letus,nurgle,lvmt}.
Every block must commit to a cryptographic digest of the entire state so that light clients can verify query results without trusting full nodes.
This authenticated commitment step, not transaction execution, has become the binding constraint on end-to-end block throughput and confirmation latency.

Execution-layer optimizations reduced the cost of everything \emph{except} commitment.
Transaction-level parallelism~\cite{occda,blockstm,schain,parallelevm}, instruction-level acceleration~\cite{forerunner,seer,mtpu,dtvm,superinstruction}, and decoupled state storage architectures~\cite{chainkv,blocklsm,letus,splitDB,solsDB,erigon,reth} all improve execution throughput but cannot amortize the cost of updating the \textit{authenticated data structure} (ADS), because commitment requires rehashing dependent node paths.
Modern execution clients such as Erigon and Reth accordingly treat the ADS as a dedicated commitment engine~\cite{erigon,reth}, making its write throughput the primary bottleneck.

\begin{table*}[t]
\centering
\small
\caption{Comparison of ADS approaches for uniformly random 256-bit keys (e.g., Keccak-256 derived storage keys in Ethereum).}
\label{tab:comparison}
\vspace{5pt}
\resizebox{\textwidth}{!}{%
\begin{tabular}{c|c|c|c|c|c}
\toprule
\textbf{Approach} & \textbf{Crypto Primitive} & \textbf{Optimization} & \textbf{Height} & \textbf{Nurgle Resistance} & \textbf{Range Proof for Trie Key} \\
\midrule
MPT~\cite{ethereum} & Standard Hash & Baseline & $\leq\max(64, O(\log_{16} N))$ & \xmark\,\,\,(low span) & \cmark\, [$O(m + \log_{16} N)$] \\
RainBlock~\cite{rainblock} & Standard Hash & Storage (DSM-Tree) & $\leq\max(64, O(\log_{16} N))$ & \xmark\,\,\,(low span) & \cmark\, [$O(m + \log_{16} N)$] \\
LMPTs~\cite{lmpts} & Standard Hash & Storage (Mem/Disk Layering) & $\leq\max(64, O(\log_{16} N))$ & \xmark\,\,\,(low span) & \cmark\, [$O(m + \log_{16} N)$] \\
Prefix MPT~\cite{chainkv} & Standard Hash & Storage (Sequential Locality) & $\leq\max(64, O(\log_{16} N))$ & \xmark\,\,\,(low span) & \cmark\, [$O(m + \log_{16} N)$] \\
DMM-Trie~\cite{letus} & Standard Hash & Storage (Delta-encoded) & $\leq\max(64, O(\log_{16} N))$ & \xmark\,\,\,(low span) & \cmark\, [$O(m + \log_{16} N)$] \\
Unified Binary~\cite{binary} & Standard Hash & SNARK-Friendly (Binary) & $\leq\max(256, O(\log_{2} N))$ & \xmark\,\,\,(low span) & \cmark\, [$O(m + \log_{2} N)$] \\
\midrule
LVMT~\cite{lvmt} & Vector Comm. & VC-based (High-span + HMT) & $\leq\max(16, O(\log_{65536} N)) + O(\log_2 (\text{Epoch} \times \Delta))$ & \cmark\, (high span) & Not Native \\
Verkle Trie~\cite{verkle} & Vector Comm. & VC-based (High-span + stem) & $\leq\max(32, O(\log_{256} N))$ & \cmark\, (high span) & Not Native \\
\midrule
\textbf{\mhot{}} & \textbf{Standard Hash} &\makecell{\textbf{Cross-layer}\\ \textbf{(Adaptive Span + Bit-indexing + SIMD)} }& $\pmb{\leq\max(\lceil\frac{256}{k-1}\rceil, O(\log_k N))}$ & \textbf{\cmark\, (adaptive span)} & \cmark\, \pmb{[$O(m + \log_{k} N)$]} \\
\bottomrule
\end{tabular}%
}
    \vspace{-0.1in}
\end{table*}

\heading{Limitations of current approaches} (Table~\ref{tab:comparison}).
The Merkle Patricia Trie (\mpt{}), Ethereum's canonical authenticated state structure~\cite{ethereum}, stores state in a Merkle-authenticated key-value trie. Lookups and updates follow the hexadecimal digits of a key from the root to a leaf, so each update rehashes every node on that path. As Ethereum's state has grown, average path depth has reached 8 to 11 levels~\cite{chainkv}; the Nurgle attack~\cite{nurgle} exploits the same prefix structure by choosing keys with long common prefixes and pushing selected paths toward the worst case. Partitioning, checkpointing, and in-memory upper-level caching, as in Chainspace~\cite{chainspace}, RainBlock~\cite{rainblock}, and LMPTs~\cite{lmpts}, improve scalability and average-case latency, but they preserve prefix-based traversal and leave worst-case height and Nurgle resistance unresolved.

Unfortunately, the limitation is structural. 

In a prefix-based trie, each tree level consumes a fixed number of contiguous key bits. We call this per-level bit-width the \emph{span} ($s$). A node that consumes $s$ bits must provide a child slot for every possible $s$-bit pattern; we call this number of child slots the \emph{fanout}. In a prefix-based trie, $\text{fanout}=2^{\text{span}}$, so the two quantities are coupled \textbf{exponentially}.

The apparent defense against adversarially shared prefixes is to increase the span so that each level consumes more key bits. However, large span inflates both node size and membership proofs, which must include sibling hashes at every level. For example, bounding the worst-case depth to at most 10 for 256-bit keys requires $s \geq 26$, resulting in roughly 67 million child slots per node. This \emph{span--proof trade-off} implies that no prefix-based scheme can simultaneously guarantee bounded worst-case height and compact membership proofs.

Vector-commitment alternatives (e.g., Verkle tries~\cite{verkle}, \lvmt{}~\cite{lvmt}) avoid this trade-off by algebraically decoupling proof size from fanout. However, both rely on trusted setup ceremonies and pairing-based elliptic-curve commitments, introducing trust assumptions and higher arithmetic costs that hash-based schemes do not require. Moreover, neither provides native support for range proofs, which are essential for verified state synchronization.

\heading{Our approach: height optimality without cryptography.}
The key observation behind our approach is that prefix indexing causes the \emph{span--proof trade-off}. This observation motivates our use of Height-Optimized Tries (HOT)~\cite{hot} as the structural foundation.

HOT branches on \emph{discriminative bits} rather than on fixed prefixes. A discriminative bit is a bit position where at least two keys in the current subtree differ. Prefix-based indexing consumes bits at fixed positions regardless of whether they distinguish any keys; discriminative-bit indexing consumes a bit only where keys actually diverge. This difference lets HOT achieve provably minimal height among radix tries for any given key set.
Discriminative-bit indexing couples span and fanout \emph{linearly}, since a compound node with fanout $k$ consumes up to $k-1$ discriminative bits, compared to the $\log_2 k$ contiguous bits resolved by a prefix-based node of the same fanout.
This linear coupling gives compound nodes enough span to absorb adversarial insertions without immediate depth growth.
When restructuring is needed, HOT's structure-adapting insertion algorithm places the new key to preserve the height-optimized invariant.
These properties prevent long shared prefixes from directly turning into long root-to-leaf paths.

\mhot{} inherits HOT's wide compound nodes, but this may  create a new authentication cost.
A direct hash-based proof for a child would include all $k-1$ sibling hashes inside the compound node, so proof size grows linearly with fanout and, in \mhot{}, with span.
We address this with \emph{hierarchical proof}, a two-layer Merkle architecture where each compound node maintains an internal Merkle tree over its children.
Proving membership of a child then requires only $O(\log k)$ sibling hashes from this internal tree. Thus, proof size grows only logarithmically with \mhot{}'s fanout and, due to HOT's linear span--fanout coupling, only logarithmically with its realized span. This is different from a prefix-indexed trie: because its fanout grows as $2^{span}$, the same internal Merkle tree would still leave proof size linear in span.

Another challenge is that HOT was originally designed for in-memory indexing and does not support persistence or authentication~\cite{hot}.
The original work explicitly identifies disk-based deployment as an open challenge, since discriminative-bit tracking breaks the prefix locality relied upon by traditional persistent tries.
Adapting HOT to \mhot{} therefore requires preserving its structural invariants while adding content-addressable persistence and cryptographically binding state roots. The search and insertion logic must remain intact, and membership, multipoint, and range proofs must be generated and verified against those roots.

We resolve these limitations in \mhot{} by making HOT-style nodes content-addressable, so node identifiers and state roots follow deterministically from the trie structure.
This design makes \mhot{} persistent and cryptographically binding, while a parallel height-stratified commit pipeline and an LSM-tree-friendly layout control commitment and storage costs.

\heading{Our contribution.} We present \mhot{}, an instantiation of HOT that supports both persistence and authenticated commitment, composing naturally with content-addressable, copy-on-write storage backends.
Beyond throughput, \mhot{}'s compound node design provides structural resistance to Nurgle attacks.
Even when an adversary controls an entire block's gas budget, prefix collisions are absorbed through internal restructuring without propagating depth increases.

Evaluation under Ethereum mainnet workloads shows \mhot{} achieves up to 9$\times$ higher write throughput and 4$\times$ lower write amplification than \mpt{}-based implementations, with 2$\times$ smaller proofs. Under Nurgle attack conditions, \mhot{} exhibits a 0\% attack success rate, in contrast to the 99.97\% success rate observed for \mpt{}.

In short, we make the following contributions:
\begin{itemize}[nosep]
\item We identify span-proof coupling as the structural reason prefix-based authenticated tries cannot easily combine low worst-case height with compact proofs. \mhot{} avoids this coupling through HOT's linear span-fanout relation and therefore uses only standard hash commitments rather than relying on vector commitments.

\item We adapt HOT to blockchain state commitment, presenting \mhot{} as a height-optimized ADS that resists adversarial key distributions through discriminative-bit tracking and high-fanout compound nodes.

\item We design \emph{hierarchical proofs}, a two-layer Merkle construction that reduces per-node proof overhead from $O(k)$ to $O(\log k)$, breaking the linear dependence of proof size on span.

\item We resolve the open problem identified in the original HOT work~\cite{hot} on disk-based deployment through content-addressable indexing, a parallel height-stratified commit pipeline, and an LSM-tree-friendly layout for persistent keys.

\item We benchmark \mhot{} against \mpt{}, RainBlock, and \lvmt{} under Ethereum mainnet workloads, demonstrating up to 9$\times$ higher throughput and 2$\times$ smaller proofs.
\end{itemize}
\section{Why Existing Solutions Fail}
\label{sec:why-fail}

Each ADS existing approach (Table~\ref{tab:comparison}) sacrifices at least one of listed properties: acceptable worst-case height, Nurgle resistance, transparent setup, or native range proofs.
The root cause is prefix-based indexing. We dig into it.

\subsection{The \mpt{} Limitation}
\label{sec:mpt}

Ethereum stores its state in a Merkle Patricia Trie (MPT). To locate a value, the trie reads the key one hex digit at a time and follows the matching child pointer downward. When two keys share a prefix, they ride the same path until their digits diverge, at which point the trie forks. Branches that only a single key traverses are not kept as individual levels. Instead, the whole stretch is compressed into one node that remembers the shared prefix. Values sit at the leaves. Formally, the forking nodes are branch nodes, the compressed stretches are extension nodes, and the value-bearing endpoints are leaf nodes. This design, a Patricia-compressed radix trie, strips out unary chains but preserves the fixed-prefix indexing rule. Every update walks a \textbf{root-to-leaf} path, so changing any entry recomputes hashes along that whole path.

Ethereum implements this trie with content-addressed, copy-on-write nodes. Each materialized node is identified by the hash of its serialized content. Modifying a key-value pair creates a new leaf representation and recomputes hashes along the path from that leaf to the root (Figure~\ref{fig:mpt-cow}). For a realized path of depth $d$, one state modification requires $O(d)$ hash computations and $O(d)$ node writes or lookups.
Measurements show that average \mpt{} depth on Ethereum mainnet has grown from 8 to 11 levels as state size increases.
A single block now triggers over 8,700 disk I/O operations~\cite{chainkv}.

This depth-dependent performance creates a denial-of-service (DoS) vector. The Nurgle attack~\cite{nurgle} exploits the \mpt{}'s prefix-based indexing and Patricia compression to deepen realized trie paths. Attackers search for keys whose hashed trie indices share long prefixes with existing keys. Because the \mpt{} routes keys by prefix, these keys traverse the same trie region before diverging. When such a key is inserted into a region where Patricia compression has collapsed a long shared prefix into one extension node, the insertion forces that node to split. A compressed path is materialized into additional branch and extension nodes, increasing the realized root-to-leaf depth by 1--2 levels per insertion. Later accesses to the affected keys require more node traversals, hash verifications, and disk lookups. Because Ethereum's gas pricing charges per opcode rather than per node traversed~\cite{nurgle}, the attacker pays almost nothing while validators absorb the I/O cost of the deepened paths. Current mitigations, including historical data pruning~\cite{eip4444} and gas repricing~\cite{eip4762}, address symptoms rather than the structural root cause.

\subsection{The Span--Proof Trade-off}
\label{sec:span-proof}

\heading{Storage and system-level mitigations.}
Several approaches reduce authenticated-state cost without changing the local indexing rule. RainBlock~\cite{rainblock} decouples storage from consensus through sharded in-memory state. LMPTs~\cite{lmpts} keep recent-update tries in memory and store larger snapshot tries on disk. Chainspace~\cite{chainspace} splits smart-contract state across shards, builds commitments for shard-local state and history, and certifies those commitments through shard-level quorum signatures. ChainKV~\cite{chainkv} exploits sequential key locality through prefix-based storage. Letus~\cite{letus} uses delta-encoded state representation.

These mitigations improve normal-case performance and can reduce the short-term cost of adversarial accesses, but they do not remove the local depth mechanism exploited by Nurgle~\cite{nurgle}. If a shard, checkpoint, or cache-backed component still uses a fixed-prefix trie, an adversary can target keys whose distinguishing prefixes fall below the cached or partitioned boundary. The attack surface narrows or shifts, but the local tree still routes by fixed prefixes.

With a fixed 16-ary branching factor, each node consumes 4 prefix bits per level, so 256-bit keys have worst-case depth 64. Storage and system-level mitigations therefore complement, rather than replace, a local authenticated structure whose height resists adversarial prefix construction.

\heading{The exponential coupling.}
Prefix-based indexing couples span and fanout exponentially (\S\ref{sec:introduction}), making high-span nodes intractable.
The exponential blowup extends to proofs and commitment computation.
A naive membership proof must include all sibling hashes at each tree level, yielding $O(d \cdot 2^s)$ hash values in the proof.
Hierarchical hashing reduces this to $O(d \cdot s)$ hashes, but each node still requires $O(2^s)$ hash computations during commitment.
Prefix-based indexing thus forces a choice: low depth with impractically large fanout, or practical node sizes with deep trees vulnerable to attack.

\begin{figure}[t!]
    \centering
    \includegraphics[width=\columnwidth]{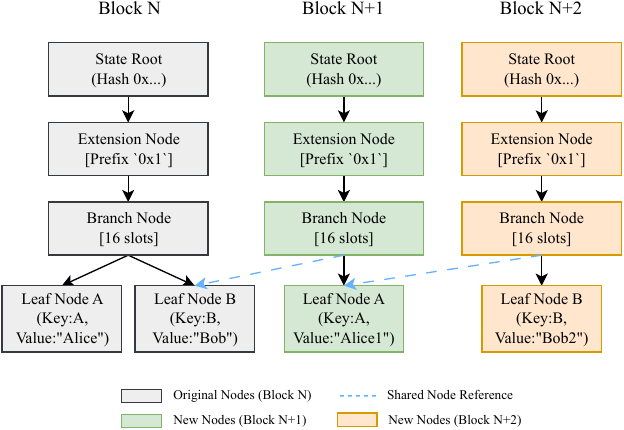}
    \vspace{-0.2in}
    \caption{Cross-block node sharing in \mpt{} under copy-on-write semantics. When accounts are modified, only affected paths change. Unchanged subtrees remain shared. Such append-only structure requires $O(d)$ writes per modification, where $d$ is the path length from leaf to root.}
    \label{fig:mpt-cow}
      \vspace{-0.2in}
\end{figure}

\heading{Cryptographic approaches.}
Vector-commitment schemes sidestep this trade-off by decoupling proof size from fanout through algebraic techniques.
Verkle trees~\cite{verkle} use inner-product arguments to achieve $O(1)$ proof size per node with 256-ary fanout.
\lvmt{}~\cite{lvmt} extends this principle to $2^{16}$-ary nodes via KZG polynomial commitments~\cite{kzg}, achieving $O(1)$ updates through version-value separation.

Both Verkle trees and \lvmt{} require trusted setup ceremonies whose compromise would enable proof forgery.
Their pairing-based verification also imposes higher costs on resource-constrained light clients than standard hashing does.
\lvmt{} inherits the Authenticated Multipoint Evaluation Tree's in-place update model and thus cannot support historical state queries~\cite{lvmt, amt}.
Neither natively supports range proofs, proofs that all keys within a specified interval satisfy certain properties, a capability essential for state synchronization protocols.

Prefix-based indexing is the shared root cause of both camps' limitations.
Storage-layer optimizations preserve transparency and range proofs but inherit Nurgle vulnerability because their low span cannot absorb adversarial insertions.
Cryptographic solutions achieve Nurgle resistance through high span but sacrifice transparency and range proofs.
Achieving attack-resistant depth without these trade-offs requires an indexing paradigm that decouples span from fanout.
\section{Key Primitives}
\label{sec:primitives}


\subsection{Authenticated Data Structures}
\label{sec:ads}

An authenticated data structure (ADS) allows an untrusted prover to certify data-operation correctness without requiring the verifier to possess the entire dataset~\cite{ads}. Formally, an ADS provides three core primitives.

\begin{packeditemize}
\item \textbf{Commit.} Given a dataset $D$, produce a short cryptographic digest $C$ that binds to $D$. In blockchain terminology, $C$ corresponds to the \emph{state root} stored in block headers.

\item \textbf{Prove.} Given a query $q$ and its result $a$, generate a proof $\pi$ demonstrating that $a$ is the correct answer for $q$ under the committed dataset.

\item \textbf{Verify.} Given the commitment $C$, query $q$, claimed result $a$, and proof $\pi$, compute $\texttt{Verify}(C, q, a, \pi) \in \{0, 1\}$ to determine whether $a$ is authentic.
\end{packeditemize}

A secure ADS has two properties.
\emph{Computational binding} makes proof forgery computationally infeasible. \emph{Succinctness} bounds proof size $|\pi|$ and verification time to $O(\log |D|)$, with some constructions achieving $O(1)$ for both~\cite{verkle, kzg}.

Consensus protocols agree only on block headers containing the state root, not the full state.
Blockchains therefore rely on ADS to make these compact commitments verifiable, enabling light clients to verify query responses against the block header rather than trusting RPC providers.

Modern execution clients (e.g., Erigon~\cite{erigon}, Reth~\cite{reth}) decouple plain state storage from ADS.
The underlying database handles execution-layer I/O directly, while the client propagates state updates to ADS for commitment maintenance.
Under this architecture, the ADS functions as a dedicated \emph{commitment engine}.
Its primary performance metric is \emph{write throughput}, measured as batch updates committed per second.

\subsection{Height-Optimized Trie}
\label{sec:hot}

\hot{}~\cite{hot} breaks exponential coupling (\S\ref{sec:span-proof}) by branching on \emph{discriminative bits} rather than contiguous prefixes. Given a set of keys $K$ within a subtree, a bit position $i$ is discriminative if at least two keys in $K$ differ at that position. The discriminative bit positions of $K$ are
\[
D(K) = \{\, i \mid \exists\, k_a, k_b \in K,\; \mathrm{bit}_i(k_a) \neq \mathrm{bit}_i(k_b) \,\}.
\]
Prefix-based indexing, in contrast, consumes a fixed block of contiguous bits at each level regardless of whether those bits distinguish any keys.

Discriminative-bit indexing couples span and fanout \emph{linearly}.
A compound node with fanout $k$ consumes up to $k - 1$ discriminative bits at arbitrary positions.
A prefix-based node with the same fanout resolves only $\log_2 k$ contiguous bits per level, requiring far more levels to cover the full key.
Concretely, a compound node with $k = 27$ entries consumes up to 26 discriminative bits, well within practical limits.

\emph{Compound nodes} aggregate multiple binary decisions into a single node containing up to $k$ entries, where $k = 32$ is typical.
Each compound node stores three components.
The \emph{extraction mask} identifies discriminative bit positions within this node's scope.
The \emph{sparse partial keys} record only the discriminative bits for each entry.
The \emph{entry data} holds either a child pointer or a leaf value.

The result is \emph{adaptive span}.
In dense key regions, a node covers few bit positions with many entries; in sparse regions, a node spans many positions with few entries.
Tree height adapts to actual key distribution rather than following fixed prefix widths.
For any key set and fanout parameter $k$, \hot{} achieves provably minimal height among radix tries~\cite{hot}.
For uniformly random 256-bit keys with $k = 32$, this yields a theoretical minimum depth of $\lceil 256/31 \rceil = 9$ levels, nearly $6\times$ fewer than the 52 levels required by a prefix-based trie with the same fanout ($\log_2 32 = 5$ bits per level).

Earlier work measures path length in \emph{depth}, while \hot{} uses  \emph{height} following~\cite{hot}; both count compound-node hops from root to leaf. \hot{} defines the height of an internal node $n$ as $h(n)=\max_i h(c_i)+1$, and $h(n)=0$ for leaves.
Unequal child heights create \emph{height gaps}, which allow new nodes to be absorbed without increasing the global tree height.

\hot{} preserves height optimality via four insertion modes (Figure~\ref{fig:hot-insertion}) that adapt the structure based on node height relations~\cite{hot}.
Only \emph{Parent Pull-Up} increases global height, while \emph{Normal Insert}, \emph{Leaf Pushdown}, and \emph{Intermediate Node Creation} perform local restructuring that may lengthen individual paths.
\S\ref{sec:eval:security} examines whether such local path increases can be exploited for Nurgle-style attacks.

\begin{figure}[t]
    \centering
    \includegraphics[width=\columnwidth]{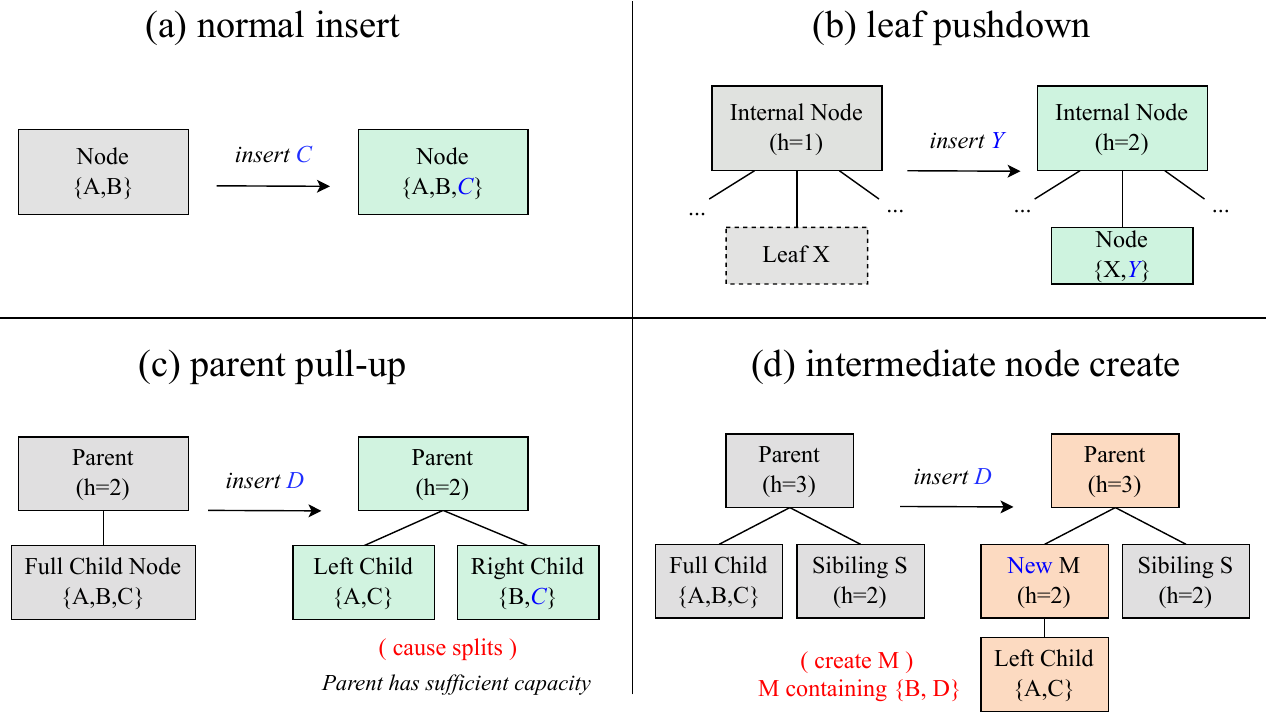}
    \vspace{-0.2in}
    \caption{\hot{} insertion mechanisms with fanout $k=3$. (a)~\emph{Normal Insert}: the new key is added when the node has capacity. (b)~\emph{Leaf pushdown}: a collision creates a new child node containing both entries. (c)~\emph{Parent pull-up}: overflow propagates upward when $h_{\text{child}}+1 = h_{\text{parent}}$; this is the only path that increases global tree height. (d)~\emph{Intermediate node creation}: when $h_{\text{child}}+1 < h_{\text{parent}}$, an intermediate node absorbs the split within the height gap.}
    \label{fig:hot-insertion}
 \vspace{-0.15in}
\end{figure}

Binna et al.~\cite{hot} prove that \hot{}'s dynamic construction produces the same structure as Static Minimum Height Partitioning (SMHP)~\cite{smhp} of the underlying binary Patricia trie.
This equivalence guarantees that \hot{} minimizes compound nodes on any root-to-leaf path.

\emph{Structural determinism} means that a given key set produces an identical trie structure regardless of insertion order.
\hot{} inherits this property from its underlying Patricia trie representation, since both the Patricia trie and SMHP partitioning are uniquely determined by the key set~\cite{hot}.

These properties establish \hot{} as the structural basis for blockchain state commitment, but the original design operates in memory and provides no persistence or authentication.
\section{Persistent Authentication for HOT}
\label{sec:design}

\mhot{} extends HOT with persistent authentication, enabling their use as a blockchain state commitment engine. We introduce a set of orthogonal extensions that build on top of \hot{}’s core algorithms without modifying them.

We first outline requirements in \S\ref{sec:design:requirements}.
We then present key components: We introduce content-addressable persistence to provide deterministic node identifiers while inheriting cryptographic authentication (\S\ref{sec:design:persistence}).
We design a two-layer Merkle construction that reduces per-node proof overhead from $O(k)$ to $O(\log k)$ (\S\ref{sec:design:authentication}).
We develop a batched commit pipeline that defers and parallelizes hash computation, eliminating redundant work when insertions share ancestors (\S\ref{sec:design:commit}).
We apply storage-layer optimizations that exploit \hot{}’s structural properties to improve disk I/O efficiency (\S\ref{sec:design:storage}).

\subsection{Design Requirements}
\label{sec:design:requirements}

A blockchain authenticated data structure must satisfy five requirements.
\emph{Short commit paths} reduce the computation per update, directly affecting block execution time.
\emph{Structural determinism} guarantees that identical key sets produce identical tree structures, a prerequisite for validators to agree on state roots.
\emph{Trust minimization} avoids trusted third parties, restricting security assumptions to cryptography alone.
\emph{Memory efficiency} reduces disk I/O by caching frequently accessed state in memory.
\emph{Attack resistance} prevents adversaries from exploiting structural flaws to extend storage paths and exhaust client I/O~\cite{eip3102, eip6800, ethereum, eip1186, eip2929, eip150, nurgle}.

Short commit paths and structural determinism follow from \hot{}'s height optimality and greedy partitioning, while memory efficiency derives from adaptive linearized node layout and SIMD-accelerated operations; we validate attack resistance empirically (\S\ref{sec:evaluation}).
Two design challenges remain.
First, the original \hot{}'s page-based node identifiers are machine-dependent and carry no cryptographic meaning; content-addressable indexing (\S\ref{sec:design:persistence}) resolves this by deriving each identifier from the hash of its content, satisfying trust minimization without a trusted setup.
Second, embedding authentication in high-fanout compound nodes risks inflating proof size; our authentication mechanism targets logarithmic overhead to keep proofs practical (\S\ref{sec:design:authentication}).

\subsection{Content-Addressable Persistence}
\label{sec:design:persistence}

Indexing \hot{} nodes demands a custom strategy; approaches that work for conventional tries fail for \hot{}'s compound-node structure.
We consider and reject two naive alternatives before arriving at our solution.

Traditional radix tries index nodes by path prefix, but \hot{}'s compound nodes track only discriminative bits without recording complete prefixes, so prefix-based indexing does not apply.
One might instead use accumulated discriminative bits as identifiers; however, \hot{}'s insertion algorithm dynamically adjusts which bits are discriminative, causing cascading identifier updates across affected subtrees.

We design \mhot{} around content-addressable indexing, where each node's database key equals the hash of its serialized content.
This scheme exploits \hot{}'s structural determinism (\S\ref{sec:hot}), which guarantees that identical key sets yield identical tree structures.
A node's content depends uniquely on the keys it covers, since extraction masks, sparse partial keys, and child references are all determined by those keys; identical logical nodes therefore produce identical hashes regardless of creation time, so all validators derive the same identifiers for the same state.
Hash-based indexing also requires no trusted setup, satisfying the trust minimization requirement.

Like \mpt{}~\cite{ethereum}, content-addressable storage inherits copy-on-write semantics (see Figure~\ref{fig:mpt-cow}), where modifying any node invalidates its hash and all ancestor references.
This approach expands child references from the original 4--8-byte pointers to 40-byte identifiers (Figure~\ref{fig:node-layout}), adding up to ${\sim}1$\,KB per fully-populated node at $k = 32$.
We accept this overhead because uniform, fixed-format references simplify serialization and eliminate branching in traversal code.

\mhot{} fixes the key length at 256 bits to match the output of Keccak-256 used for Ethereum’s storage key derivation. As a result, extraction masks and sparse partial keys have fixed maximum sizes and are stored in fixed-length fields, while only the child reference array grows dynamically with the number of entries. This layout enables simple, single-pass serialization. Figure~\ref{fig:node-layout} illustrates the node structure.

\begin{figure}[t]
\centering
\includegraphics[width=0.8\columnwidth]{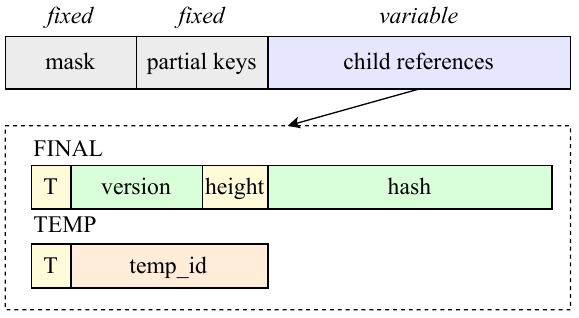}    \vspace{-0.1in}
\caption{Node layout in \mhot{}. The top panel shows a compound node with fixed-size metadata and variable-length child array. The bottom panel shows two child reference formats. \textsc{Final} references store version, height, and hash; \textsc{Temp} references store a temporary ID for deferred hashing.}
\label{fig:node-layout}
\end{figure}

\subsection{Hierarchical Proof}
\label{sec:design:authentication}

Each compound node embeds $k$ child hashes, so proof size grows linearly with fanout.
A naive inclusion proof must contain all $k-1$ sibling hashes at each level so the verifier can reconstruct parent hashes up to the root, yielding $O(h \cdot k)$ hashes per proof.
With 32-byte Keccak-256 hashes, this totals $h(k{-}1) \cdot 32 \approx 5$\,KB for typical parameters ($k = 32$, $h = 5$).

We address this using a two-layer Merkle architecture that trades additional commitment-time computation for compact proofs. 
Each compound node computes a \emph{Children Merkle Root} (CMR), a binary Merkle root over its $k$ child hashes. 
Membership proofs then include only $O(\log k)$ binary Merkle siblings per node, instead of all $k-1$ child hashes, reducing per-node overhead from $O(k)$ to $O(\log k)$.
For $k=32$, this yields roughly a $6\times$ reduction in hash overhead.
In practice, single-point proofs are about 1.1--1.4\,KB, compared to 2.3--2.9\,KB for MPT (\S\ref{sec:evaluation}).

We adopt binary Merkle trees to preserve standard hash-based verification without trusted setup.
Although each CMR incurs $O(k)$ hashing cost at commit time, this overhead is absorbed by the batched commit pipeline (\S\ref{sec:design:commit}).
\S\ref{sec:proofs:architecture} formalizes the construction and extends it to compact multiproofs.

\subsection{Batch Commit Pipeline}
\label{sec:design:commit}

Each insertion rehashes the entire root-to-leaf path. When multiple insertions share ancestors, naive processing repeatedly rehashes the same \emph{dirty nodes}, i.e., nodes modified since the last commit. We eliminate this redundancy by proposing three optimizations as below.

\heading{Deferred hashing.}
We defer hashing until block commit using \emph{mixed child references}.
Each reference stores a type tag distinguishing finalized hashes from temporary IDs.
During insertion, dirty nodes receive temporary IDs and enter a pending map.
Parents reference these temporary IDs rather than hashes, deferring computation until commit.

Copy-on-write occurs at most once per node per block.
Once a node enters the pending map with a temporary ID, subsequent traversals update it in place rather than cloning.
When an insertion first modifies a finalized node, we clone it into the pending map (Algorithm~\ref{alg:insert}, lines 11--12); subsequent insertions that traverse the same node find a temporary reference and update it in place without cloning (line 9).
The resulting memory overhead is $O(h)$ pending nodes per insertion, where $h$ is tree height.

\begin{algorithm}[t]
\caption{Insert with Deferred Hashing}
\label{alg:insert}
\begin{algorithmic}[1]
\Require key $K$, value $v$, tree $T$, pending map $P$
\Ensure A temporary ID referencing the updated root
\State $P[\textsc{Alloc}()] \gets (\textsc{Leaf}(K, v), 0)$; $\mathit{cid} \gets$ allocated ID
\State $\mathit{path} \gets \langle\rangle$; $r \gets T.\mathit{root}$
\While{$r$ is internal}
    \State $n \gets \textsc{Get}(r)$; $s \gets \textsc{Slot}(n, K)$ \Comment{\textcolor{teal}{extract discriminative bits, match partial keys}}
    \State $\mathit{path}.\textsc{Push}((n, s, r))$; $r \gets n.\mathit{ch}[s]$
\EndWhile
\ForAll{$(n, s, r) \in \mathit{path}$ in reverse}
    \If{$r.\mathit{tag} = \textsc{Temp}$} \Comment{\textcolor{teal}{in $P$: update in place}}
        \State $P[r.\mathit{id}].n.\mathit{ch}[s] \gets \mathit{cid}$; $\mathit{cid} \gets r.\mathit{id}$
    \Else \Comment{\textcolor{teal}{finalized: clone once}}
        \State $n' \gets \textsc{Clone}(n)$; $n'.\mathit{ch}[s] \gets \mathit{cid}$
        \State $P[\textsc{Alloc}()] \gets (n', n'.\mathit{ht})$; $\mathit{cid} \gets$ allocated ID
    \EndIf
\EndFor
\State \Return $\mathit{cid}$
\end{algorithmic}
\end{algorithm}

\heading{Height-stratified batch commit.}
At block boundaries, we finalize all pending nodes in one batched traversal (cf. Algorithm~\ref{alg:commit}).
Correct ordering requires processing children before parents.
\hot{} derives this order from each node's height (recall \S\ref{sec:hot}); since children have strictly lower heights, ascending height order satisfies the dependency.
We group nodes by height and process from leaves upward, replacing temporary references with finalized hashes at each level.

\heading{Parallel execution.}
Nodes at the same height share no data dependencies, allowing concurrent hash computation within each level.
Processing advances to the next height only after synchronizing all hashes at the current height.
Correctness follows because temporary references always point to children at strictly lower heights, so each level reads only from entries that prior levels finalized in $F$ (cf. Algorithm~\ref{alg:commit}).

Because \hot{} records node height as a built-in property, \mhot{} can use height directly to schedule commit. Nodes at the same height have no Merkle dependencies on one another, so they can be hashed in parallel with one barrier between levels, without fine-grained synchronization or uneven subtree splits. The same height grouping aligns with the storage layout in \S\ref{sec:design:storage}, allowing finalized nodes to be flushed in locality-preserving batches. \mpt{}-style tries do not expose this height order directly, so recovering such a schedule requires extra bookkeeping~\cite{geth-parallel-node-fetching}.

\begin{algorithm}[t]
\caption{Height-Stratified Batch Commit}
\label{alg:commit}
\begin{algorithmic}[1]
\Require pending map $P$: ID $\to$ (node, height), version $v$
\Ensure Root reference $(v, \mathit{hash})$
\State $F \gets \emptyset$ \Comment{\textcolor{teal}{temp ID $\to$ finalized ref}}
\For{$h \gets 0$ \textbf{to} $\max(P.\mathit{heights})$}
    \ForAll{$(\mathit{id}, n) \in P$ at height $h$} \textbf{in parallel}
        \ForAll{$c$ in $n.\mathit{ch}$}
            \If{$c.\mathit{tag} = \textsc{Temp}$} $c \gets F[c.\mathit{id}]$
            \EndIf
        \EndFor
        \State $F[\mathit{id}] \gets (v, \textsc{Hash}(\textsc{Ser}(n)))$
    \EndFor
\EndFor
\State \Return $F[\mathit{root\_id}]$
\end{algorithmic}
\end{algorithm}

\subsection{Storage Optimizations}
\label{sec:design:storage}

Embedding hashes enlarges node footprint; three techniques mitigate this overhead.

\heading{Metadata compression.}
We serialize only populated fields.
The extraction mask stores positions of set bits rather than a full 256-bit vector; sparse partial keys use adaptive bit widths.
All fields concatenate without internal pointers, so serialization completes in a single pass.

\heading{Write ordering.}
We batch writes across multiple blocks.
Each block computes its state root immediately, but nodes accumulate in memory until a configurable flush threshold triggers a single batch write.
Each database key consists of a 64-bit prefix followed by the 256-bit content hash, totaling 40 bytes.
The prefix encodes a 1-bit type flag distinguishing leaves from internal nodes, a 55-bit version number recording the block of creation, and an 8-bit height value.
Version-prefixed ordering groups nodes by creation block, while the height suffix clusters same-height nodes within each block.
This layout aligns with \mhot{}'s height-stratified commit and enables sequential batch reads during traversal.
Sorting keys before writing yields locality that LSM-tree storage engines exploit for efficient ingestion~\cite{chainkv,lsmtree}.

\heading{Asynchronous persistence.}
We decouple state root computation from disk I/O by delegating writes to a background thread, avoiding stalls during block execution.
When multiple batches complete before a flush finishes, pending writes coalesce into a single disk operation, amortizing \texttt{fsync} overhead.
The underlying LSM-tree engine's write-ahead log ensures durability; committed data remains recoverable even if a crash precedes the next flush.

The next section formalizes \mhot{}'s proof mechanisms (\S\ref{sec:proofs}), followed by performance evaluation (\S\ref{sec:evaluation}).
\section{Proof Mechanisms}
\label{sec:proofs}

\mhot{} supports four proof types within a unified framework: \textit{single-point membership} and \textit{non-membership} proofs, \textit{multi-point membership proofs}, \textit{lower bound proofs}, and \textit{range proofs}.
This section formalizes each mechanism with proof sketches; Appendix~\ref{sec:proofs-formal} gives the full game-based security proofs.
\subsection{Two-Layer Merkle Architecture}
\label{sec:proofs:architecture}

The naive approach of including all children's hashes directly in each proof node yields proof sizes of $O(h \cdot k)$, where $h$ is tree height and $k$ is fanout.
For typical parameters ($k=32$, $n=10^8$), this produces approximately 5KB per proof---unacceptable for bandwidth-constrained applications.

\mhot{} addresses this through a two-layer Merkle architecture.
The \emph{inter-node layer} maintains the tree structure connecting compound nodes.
The \emph{intra-node layer} introduces a Merkle tree~\cite{merkle} over children within each compound node, enabling logarithmic-size proofs for child membership.

\begin{definition}[Children Merkle Root]
\label{def:cmr}
For a compound node $N$ with children $C = (c_0, \ldots, c_{|N|-1})$, the children Merkle root is defined as:
\begin{equation}
\mathrm{CMR}(N) = \mathrm{MerkleRoot}(H(c_0), H(c_1), \ldots, H(c_{|N|-1}))
\end{equation}
where the Merkle tree is padded to the next power of 2 using a canonical zero hash.
\end{definition}

\begin{definition}[Node Content Hash]
\label{def:node-hash}
The node content hash of a compound node $N$ incorporates all structural information:
\begin{equation}
H_{\text{content}}(N) = H(M \parallel S \parallel \mathrm{CMR}(N) \parallel L)
\end{equation}
where $M$ denotes extraction masks (32 bytes), $S$ denotes sparse partial keys ($|N| \times 4$ bytes), and $L$ denotes child leaf counts ($|N| \times 4$ bytes).
\end{definition}

This architecture reduces per-node overhead from $O(k)$ to $O(\log k)$ hashes.

\begin{definition}[Node Proof Entry]
\label{def:npe}
For a compound node $N$ and child index set $J \subseteq [0, |N|)$ with $t = |J|$:
\begin{equation}
\mathrm{NPE}(N, J) = \bigl(\, J,\; M,\; S,\; L,\; \eta,\; v,\; \underbrace{\Pi_J^{\mathrm{CMR}}}_{\text{intra-node proof}} \,\bigr)
\end{equation}
where $\Pi_J^{\mathrm{CMR}}$ is a compact Merkle multiproof for children $\{c_j : j \in J\}$ within $\mathrm{CMR}(N)$, requiring $O(t \cdot (\log k - \log t))$ sibling hashes~\cite{merkle-multiproof}.
When $t = 1$, this reduces to a standard Merkle proof with $O(\log k)$ siblings.
\end{definition}

All proof types discussed in the following embed $\Pi_J^{\mathrm{CMR}}$ within each node entry, ensuring proof size scales with $O(\log k)$ rather than $O(k)$ per node.

\subsection{Single-Point Proofs}
\label{sec:proofs:single}

\mhot{}'s single-point proofs rely on the fact that \hot{} search constitutes an \emph{optimistic search}.
Unlike traditional Patricia tries where searching for a non-existent key may fail mid-traversal, \hot{}'s sparse matching semantics guarantee that search always reaches some leaf---though this leaf may differ from the query key.
The search only matches discriminative bits at each node; the reached leaf shares all discriminative bit values with the query but is not guaranteed to be lexicographically adjacent.
A final comparison between the reached leaf and the query key determines membership.

\begin{lemma}[HOT Optimistic Search Invariant]
\label{lem:optimistic}
For any non-empty \hot{} trie $T$ and any query key $K$, the optimistic search procedure terminates at exactly one leaf node $K'$ that agrees with $K$ on all discriminative bits encountered during traversal, regardless of whether $K$ exists in $T$.
\end{lemma}

\begin{proof}[Proof sketch (Lemma~\ref{lem:optimistic})]
\hot{} search at each internal node computes $\mathrm{dense}(K, M)$ and finds the last index $j$ satisfying $(\mathrm{dense} \land \mathrm{sparse}[j]) = \mathrm{sparse}[j]$.
By the \hot{} construction invariant, the first sparse partial key is always 0 (corresponding to the leftmost subtree).
Since $(\mathrm{dense} \land 0) = 0$ holds for all dense values, at least one match exists at every non-empty node.
The search thus proceeds deterministically to a unique leaf that matches the query on all discriminative bits.
However, non-discriminative bits may differ, so the final leaf $K'$ must be compared against $K$ to determine membership.
\end{proof}

This invariant unifies membership and non-membership proofs as they both execute identical search algorithms, differing only in the final comparison.

\begin{corollary}[Proof Unification]
\label{cor:unification}
Membership and non-membership proofs share identical proof generation and path verification algorithms.
The distinction lies only in the final predicate:
\begin{equation}
\mathrm{ProofType}(K, K_0) =
\begin{cases}
\text{Membership} & \text{if } K_0 = K \\
\text{Non-Membership} & \text{if } K_0 \neq K
\end{cases}
\end{equation}
where $K_0$ denotes the key of the leaf reached by searching $K$.
\end{corollary}

\heading{Membership proof.}
For a key $K$ with associated value $V$, the membership proof takes the form:
\begin{equation}
\pi_{\text{mem}}(K) = (K, V, v_{\text{leaf}}, \text{Path})
\end{equation}
where $\text{Path} = (\text{NPE}(N_0, J_0), \ldots, \text{NPE}(N_{h-1}, J_{h-1}))$ traces the search path from root to leaf, with $|J_i| = 1$ for single-point proofs.

Algorithm~\ref{alg:membership-verify} presents the verification procedure.

\begin{algorithm}[t]
\caption{Membership Proof Verification}
\label{alg:membership-verify}
\begin{algorithmic}[1]
\Require proof $\pi$, expected root $R$
\Ensure $\textsc{True}$ if proof is valid
\PhaseComment{\textcolor{teal}{Phase 1: Verify routing consistency}}
\ForAll{entry $\in$ $\pi.\text{path}$}
    \State $d \gets \textsc{DenseKey}(\pi.\text{key}, \text{entry}.\text{masks})$
    \State $j \gets \textsc{SearchSparse}(d, \text{entry}.\text{sparse\_keys})$
    \If{$j \neq \text{entry}.\text{child\_idx}$}
        \State \Return \textsc{False} \Comment{Routing mismatch}
    \EndIf
\EndFor
\PhaseComment{\textcolor{teal}{Phase 2: Reconstruct hashes bottom-up}}
\State $h_{\text{child}} \gets H_{\text{leaf}}(\pi.\text{key}, \pi.\text{value})$
\ForAll{entry $\in$ $\pi.\text{path}.\textsc{Reversed}()$}
    \State $\text{cmr} \gets \textsc{ReconstructCMR}(h_{\text{child}}, \text{entry})$
    \State $h_{\text{child}} \gets H(\text{entry}.\text{masks} \parallel \text{entry}.\text{sparse\_keys}$
    \State \hspace{3em} $\parallel\ \text{cmr} \parallel \text{entry}.\text{leaf\_counts})$
\EndFor
\PhaseComment{\textcolor{teal}{Phase 3: Compare with expected root}}
\State \Return $h_{\text{child}} = R.\text{content\_hash}$
\end{algorithmic}
\end{algorithm}

\heading{Non-membership proof.}
For a key $K$ not present in the trie, the non-membership proof includes the neighbor leaf $(K', V')$ reached by executing the same search algorithm.
Verification additionally checks that both $K$ and $K'$ route through identical children at every node, ensuring $K'$ is indeed the leaf that \hot{} search would reach for $K$.

\begin{theorem}[Single-Point Soundness]
\label{thm:single-soundness}
If verification succeeds, then with overwhelming probability $1 - \mathrm{negl}(\lambda)$:
\begin{enumerate}[nosep]
    \item For membership proofs: $(K, V) \in T$
    \item For non-membership proofs: $K \notin T$
\end{enumerate}
\end{theorem}

\begin{proof}[Proof sketch (Theorem~\ref{thm:single-soundness})]
By collision resistance of $H$, the bottom-up hash reconstruction produces a unique sequence of node hashes.
The routing consistency check ensures $K$ traverses exactly the claimed path.
For the final hash to match the root commitment, either the path authentically exists or the adversary found a hash collision.
The latter occurs with probability at most $\mathrm{negl}(\lambda)$.
For non-membership, Lemma~\ref{lem:optimistic} guarantees that the optimistic search uniquely determines the leaf $K'$ reached for query $K$, proving no other leaf could be found by the same search procedure.
\end{proof}

\subsection{Multi-Point Membership Proof}
\label{sec:proofs:multi}

A naive approach to verifying $m$ membership proofs requires $O(m \cdot h \cdot (k-1))$ sibling hashes.
\mhot{} reduces this through two optimizations: \emph{path sharing} (keys traversing the same node share that node's metadata) and \emph{compact multiproofs} (proving $t$ children within a node requires $O(t(\log k - \log t))$ sibling hashes rather than $O(t(k-1))$).

\begin{definition}[Compact Merkle Multiproof]
\label{def:multiproof}
Given a Merkle tree $T$ with leaves $L$ and a subset $I \subseteq [0, |L|)$ of indices to prove, a compact multiproof consists of:
\begin{equation}
\pi_{\text{multi}} = (I, \Sigma)
\end{equation}
where $\Sigma$ represents the minimal set of sibling hashes needed to reconstruct the root from proven leaves $\{L[i] : i \in I\}$.
\end{definition}

Compact multiproofs exploit shared ancestors among the $m$ proven leaves to eliminate redundant sibling hashes, reducing proof size from the naive $m \cdot h$ to as few as $h + m - 1$ when leaves are clustered~\cite{merkle-multiproof}.

The compact multiproof generation algorithm (Algorithm~\ref{alg:multiproof-gen} traverses the tree bottom-up, collecting sibling hashes only for nodes where one child is known but the other is not.

\heading{Multi-point HOT proof.}
For a set of keys $K = \{K_1, \ldots, K_m\}$:
\begin{equation}
\pi_{\text{multi}}(K) = (\text{Entries}, \text{Levels})
\end{equation}
where Entries contains tuples $(K_i, V_i, v_i)$ sorted by key, and Levels organizes NPEs by depth with compact multiproofs at each node.

\begin{theorem}[Multi-Proof Soundness]
\label{thm:multi-soundness}
If verification succeeds, then all key-value pairs $(K_i, V_i)$ exist in the trie committed by the root, with overwhelming probability.
\end{theorem}

\begin{proof}[Proof sketch (Theorem~\ref{thm:multi-soundness})]
Each key $K_i$ must pass the same routing consistency and hash reconstruction checks as in single-point verification (Theorem~\ref{thm:single-soundness}).
Path sharing and compact multiproofs reduce proof size but do not weaken security: shared nodes are verified once with the same rigor, and multiproofs authenticate the same child hashes as independent proofs would.
Forging any $(K_i, V_i)$ requires finding a hash collision.
\end{proof}

\subsection{Lower Bound Proof}
\label{sec:proofs:lowerbound}

Lower bound queries---finding the smallest key $\geq Q$---are fundamental to range operations in authenticated storage systems.
Unlike membership queries where the search target either exists or does not, lower bound queries must locate a key that may differ from the query itself.

In prefix-organized tries such as MPT, siblings represent lexicographically adjacent key ranges, making neighbor identification straightforward.
\hot{} organizes nodes by discriminative bits instead, so siblings may span non-contiguous ranges. We therefore construct lower bound proofs by authenticating the search path that \hot{}'s bit-comparison traversal follows to locate the smallest key $\geq Q$.

\begin{definition}[Lower Bound Query]
\label{def:lowerbound}
For a query key $Q$, the lower bound operation $\mathrm{lb}(Q)$ returns the smallest key $K \geq Q$ present in the trie, or $\bot$ if no such key exists.
\end{definition}

A critical subtlety arises from the nature of optimistic search.
The search finds a leaf $K'$ that matches the query $Q$ on all discriminative bits encountered during traversal, but $K'$ need not be lexicographically close to $Q$---the two may differ arbitrarily on non-discriminative bits.
When searching for $Q$ terminates at leaf $K'$ where $K' \neq Q$, let $d = \mathrm{diffbit}(Q, K')$ denote the first differing bit position.
Three cases arise:
\begin{packeditemize}
    \item \emph{Exact match} ($K' = Q$): The query key exists; $K'$ is trivially the lower bound.
    \item \emph{Overshot} ($Q[d] = 0$, $K'[d] = 1$): The search entered a right subtree, so $Q < K'$ lexicographically. However, $K'$ may not be minimal within this subtree; the true lower bound is the leftmost leaf reachable from the fork point.
    \item \emph{Undershot} ($Q[d] = 1$, $K'[d] = 0$): The search ended in a left subtree where all keys are less than $Q$. The algorithm must examine right siblings at the fork point to find a subtree containing keys $\geq Q$.
\end{packeditemize}

The lower bound proof must authenticate both the leaf $K'$ reached by optimistic search and the path to the actual result.

\begin{definition}[Lower Bound Proof]
\label{def:lb-proof}
For query $Q$, the lower bound proof structure is:
\begin{equation}
\pi_{\text{lb}}(Q) = (Q, \text{Path}, K', V', v', \text{Adj}, K_{\text{result}}, V_{\text{result}})
\end{equation}
comprising the query key, the authenticated search path to the leaf $(K', V', v')$ reached by optimistic search, optional adjustment information $\text{Adj}$ for non-exact matches, and the actual lower bound result $(K_{\text{result}}, V_{\text{result}})$.
\end{definition}

When $K' \neq Q$, the adjustment information specifies how to navigate from the search path to the true lower bound:
\begin{equation}
\text{Adj} = (d, b, f, \text{AdjPath})
\end{equation}
where $d = \mathrm{diffbit}(Q, K')$ is the first differing bit, $b = Q[d]$ indicates overshot ($b=0$) or undershot ($b=1$), $f$ is the fork depth, and $\text{AdjPath}$ authenticates the path from the fork point to the result leaf.

\begin{definition}[Fork Depth]
\label{def:fork-depth}
The fork depth $f$ identifies where the query's hypothetical path would diverge from the path to the leaf reached by optimistic search:
\begin{equation}
f = \max\{i : d \in \mathrm{disc\_bits}(\text{Path}[i])\}
\end{equation}
where $\mathrm{disc\_bits}(\cdot)$ extracts the set of discriminative bit positions encoded in a node's extraction masks.
\end{definition}

\begin{lemma}[Lower Bound Correctness]
\label{lem:lb-correct}
The adjustment algorithm correctly computes $\mathrm{lb}(Q)$.
\end{lemma}

\begin{proof}[Proof sketch (Lemma~\ref{lem:lb-correct})]
Let $d = \mathrm{diffbit}(Q, K')$ denote the first bit position where the query and the reached leaf differ.

\emph{Overshot case} ($Q[d] = 0$, $K'[d] = 1$): At fork depth $f$, the search entered a subtree rooted at a node whose discriminative bit at position $d$ directed the search rightward.
All keys $K$ in this subtree share bit $d = 1$, hence satisfy $K > Q$ lexicographically.
The minimum key in this subtree (found by repeatedly taking the leftmost child from the fork point) is therefore the smallest key exceeding $Q$.
Note that this minimum need not be $K'$ itself, as $K'$ may reside anywhere within the subtree.

\emph{Undershot case} ($Q[d] = 1$, $K'[d] = 0$): The search terminated in a left subtree where all keys share bit $d = 0$, hence are lexicographically smaller than $Q$.
The algorithm must find the first right sibling at the fork point whose subtree contains keys $\geq Q$, then descend to that subtree's minimum.
\end{proof}

\heading{Verification.}
The verifier must ensure that the claimed result is indeed the minimum key $\geq Q$, not merely some key satisfying the bound.
Four checks enforce correctness:
\begin{packeditemize}
\item \emph{Path integrity.} The search path must authenticate $K'$ against the committed root via bottom-up hash reconstruction.
\item \emph{Search consistency}: Both $Q$ and $K'$ must route identically through each path node---i.e., produce the same dense key and thus select the same child---confirming $K'$ is the leaf that optimistic search reaches for $Q$.
\item \emph{Fork depth correctness.} The verifier independently recomputes $f$ from the Merkle-committed extraction masks, rejecting any mismatch with the claimed value.
\item \emph{Structural minimality.} For overshot cases, every entry in $\text{AdjPath}$ must have child index 0 (leftmost descent); for undershot cases, the first adjustment entry must be the immediate right sibling at the fork point. These constraints ensure the result is minimal.
\end{packeditemize}

\begin{theorem}[Lower Bound Soundness]
\label{thm:lb-soundness}
If verification succeeds, then $K_{\text{result}} = \mathrm{lb}(Q)$ with overwhelming probability.
\end{theorem}

\begin{proof}[Proof sketch (Theorem~\ref{thm:lb-soundness})]
The verifier independently recomputes the fork depth $f$ from the Merkle-committed extraction masks using Definition~\ref{def:fork-depth}, preventing adversarial manipulation of the branching point.
Structural constraints enforce minimality: in the overshot case, each entry in $\text{AdjPath}$ must have child index 0 (leftmost descent); in the undershot case, the first adjustment entry must be the immediate right sibling of the search path's child at the fork point.
The adjustment path must reconstruct to the same node content hash as the search path at depth $f$, cryptographically binding it to the committed trie structure. Any attempt to return a non-minimal key will produce a root hash mismatch.
\end{proof}

\subsection{Range Proof}
\label{sec:proofs:range}

Range queries are essential for authenticated storage applications such as blockchain state synchronization, where clients must verify that a returned dataset contains \emph{exactly} the entries within specified bounds~\cite{devp2p,rangeproofimportant}. The central challenge is \emph{completeness}: an adversarial prover might return a subset of the true range, omitting entries to deceive the verifier.

\begin{definition}[Range Query]
\label{def:range}
For an interval $[\text{first}, \text{last})$, the range query returns all key-value pairs $(K, V)$ satisfying $\text{first} \leq K < \text{last}$.
\end{definition}

\mhot{} constructs range proofs by composing lower bound proofs with multi-point proofs.
The lower bound proofs (Section~\ref{sec:proofs:lowerbound}) establish the range boundaries, while the multi-point proof (Section~\ref{sec:proofs:multi}) authenticates all entries within.

\begin{definition}[HOT Range Proof]
\label{def:range-proof}
For interval $[\text{first}, \text{last})$:
\begin{equation}
\pi_{\text{range}} = (\text{first}, \text{last}, \pi_{\text{lb}}^L, \pi_{\text{lb}}^R, \pi_{\text{multi}})
\end{equation}
where $\pi_{\text{lb}}^L$ authenticates $\mathrm{lb}(\text{first})$, $\pi_{\text{lb}}^R$ authenticates $\mathrm{lb}(\text{last})$, and $\pi_{\text{multi}}$ proves membership of all entries in the range.
\end{definition}

\heading{Rank-based completeness verification.}
Efficient completeness verification builds on the \emph{rank function}, which counts the number of keys preceding a given key in the trie’s total order. Since each node records the leaf counts of its child subtrees and commits them in the node content hash, the verifier can derive ranks directly from the boundary proofs, without enumerating all keys in the queried range.

\begin{definition}[Rank]
\label{def:rank}
For a key $K$ in trie $T$, the rank is the count of smaller keys:
\begin{equation}
\mathrm{rank}(K) = |\{K' \in T : K' < K\}|
\end{equation}
\end{definition}

\begin{lemma}[Rank Computation from Path]
\label{lem:rank}
Given the search path for key $K$, the rank can be computed as:
\begin{equation}
\mathrm{rank}(K) = \sum_{i=0}^{h-1} \sum_{j=0}^{\text{child\_idx}_i - 1} \mathrm{lc}(\text{path}[i].\text{children}[j])
\end{equation}
\end{lemma}

\begin{proof}[Proof sketch (Lemma~\ref{lem:rank})]
At each level $i$ of the search path, all children with indices $j < \text{child\_idx}_i$ contain keys lexicographically smaller than $K$, since sparse partial keys maintain sorted order within each compound node.
The $\mathrm{lc}$ field, committed in the node content hash, records the total number of leaves in each child's subtree.
Summing these counts across all path levels yields the total number of keys preceding $K$.
\end{proof}

This rank computation enables $O(1)$ completeness verification: the verifier simply checks whether $|\text{entries}| = \mathrm{rank}(\text{last}) - \mathrm{rank}(\text{first})$.
If an adversary omits even a single entry, the count mismatch triggers rejection.
Algorithm~\ref{alg:range-verify} presents the complete verification procedure.

\begin{algorithm}[t]
\caption{Range Proof Verification}
\label{alg:range-verify}
\begin{algorithmic}[1]
\Require proof $\pi$, expected root $R$
\Ensure $\textsc{True}$ if proof is valid
\PhaseComment{\textcolor{teal}{Verify boundary proofs first (prevents omission attacks)}}
\If{\textbf{not} $\textsc{VerifyLB}(\pi.\pi_{\text{lb}}^L, R)$} \Return \textsc{False}
\EndIf
\If{\textbf{not} $\textsc{VerifyLB}(\pi.\pi_{\text{lb}}^R, R)$} \Return \textsc{False}
\EndIf
\If{$\pi.\text{first} \geq \pi.\text{last}$}
    \State \Return $\pi.\text{entries} = \emptyset$
\EndIf
\If{\textbf{not} $\textsc{VerifyMulti}(\pi.\pi_{\text{multi}}, R)$} \Return \textsc{False}
\EndIf
\PhaseComment{\textcolor{teal}{Rank-based count verification}}
\State $r_L \gets \textsc{ComputeRank}(\pi.\text{first}, \pi.\pi_{\text{lb}}^L)$
\State $r_R \gets \textsc{ComputeRank}(\pi.\text{last}, \pi.\pi_{\text{lb}}^R)$
\If{$r_R - r_L \neq |\pi.\text{entries}|$}
    \State \Return \textsc{False} \Comment{Omission detected}
\EndIf
\PhaseComment{\textcolor{teal}{Verify ordering and boundaries}}
\ForAll{$i \in [1, |\pi.\text{entries}|)$}
    \If{$\pi.\text{entries}[i].\text{key} \leq \pi.\text{entries}[i-1].\text{key}$}
        \State \Return \textsc{False}
    \EndIf
\EndFor
\State \Return \textsc{True}
\end{algorithmic}
\end{algorithm}

The verification algorithm enforces a critical security invariant: boundary proofs must be verified \emph{before} accepting any entries.
This ordering prevents empty-proof attacks where an adversary provides valid but incomplete entry lists.
The rank-based count check then ensures exactly the correct number of entries appears.

\begin{theorem}[Range Proof Soundness]
\label{thm:range-soundness}
If verification succeeds, the entries contain exactly all keys $K$ satisfying $\text{first} \leq K < \text{last}$.
\end{theorem}

\begin{proof}[Proof sketch (Theorem~\ref{thm:range-soundness})]
The proof composes three security guarantees.
First, by Theorem~\ref{thm:lb-soundness} (Lower Bound Soundness), the boundary proofs $\pi_{\text{lb}}^L$ and $\pi_{\text{lb}}^R$ correctly identify $\mathrm{lb}(\text{first})$ and $\mathrm{lb}(\text{last})$.
Second, by Lemma~\ref{lem:rank} (Rank Computation), the verifier accurately computes $\mathrm{rank}(\text{first})$ and $\mathrm{rank}(\text{last})$ from these authenticated paths; since $\mathrm{rank}(\text{last}) - \mathrm{rank}(\text{first})$ equals exactly $|\{K : \text{first} \leq K < \text{last}\}|$, any count mismatch with $|\text{entries}|$ reveals an omission attack.
Third, by Theorem~\ref{thm:multi-soundness} (Multi-Proof Soundness), every entry in $\pi_{\text{multi}}$ authentically exists in the committed trie.
These guarantees ensure the returned entries are the keys in $[\text{first}, \text{last})$.
\end{proof}
\section{Evaluation}
\label{sec:evaluation}

We evaluate \mhot{} along four dimensions (formed by questions).
\textbf{Q1.} Does \mhot{} improve write throughput compared to existing authenticated data structures?
\textbf{Q2.} Does \mhot{}'s batched persistence strategy reduce write amplification (WA)?
\textbf{Q3.}  How does tree height vary across workloads, and what are the implications for proof size?
\textbf{Q4.}  Does \mhot{}'s compound node design mitigate Nurgle attacks?

\subsection{Experimental Setup}
\label{sec:eval:setup}

\heading{Implementation.}
We implement all systems in Rust with release optimizations.
RocksDB~\cite{rocksdb} serves as the underlying key-value store with a 2\,GB LRU cache.
We report the median of five independent runs; ranges in tables indicate variation across scales or configurations.

\heading{Hardware.}
We run experiments on an AWS EC2 instance with an 8-vCPU Intel Xeon Scalable processor (Sapphire Rapids, 3.2\,GHz), 64\,GB RAM, and EBS-optimized storage providing baseline 12,000 IOPS with burst capacity up to 40,000 IOPS.

\heading{Baselines.}
We compare \mhot{} against three representative systems under the benchmark setup used by \lvmt{}~\cite{lvmt}.  (i) \textit{\mpt{}} is Ethereum’s current authenticated state structure based on the Merkle Patricia Trie~\cite{ethereum}. (ii) \textit{\lvmt{}} is a layered versioned multipoint trie that leverages KZG polynomial commitments to achieve $O(1)$ root updates~\cite{lvmt}.  (iii) RainBlock adopts DSM-TREE, a distributed sharded Merkle tree design optimized for in-memory storage~\cite{rainblock}. The RainBlock-style configuration keeps the upper six levels in memory and pages deeper nodes from RocksDB, matching \lvmt{}'s layered-storage baseline~\cite{lvmt}. For \lvmt{}, we use the recommended configuration of 16 bits per level, resulting in a fanout of $2^{16}=65{,}536$ per layer. We do not evaluate its History Merkle Tree functionality, as this component is not available in the open-source implementation.

\heading{Workloads.}
We evaluate two types of workloads. (i) The \emph{synthetic} workload first populates the tree with 100k--1M entries, followed by 100 epochs, each consisting of 100{,}000 random updates. (ii) The \emph{real-world trace} workload replays Ethereum mainnet blocks 13{,}500{,}000--13{,}510{,}000, grouped into 200 epochs of 50 blocks each, in accordance with \lvmt{}'s recommended configuration.

\subsection{Write Throughput (Figure~\ref{fig:throughput})}
\label{sec:eval:throughput}

Write throughput measures how fast the commitment engine processes batch updates (\S\ref{sec:ads}). We denote \mhot{} with asynchronous flush as \mhot{}-AF (asynchronous flush).

\heading{Synthetic workloads.}
At 100k keys, \mhot{}-AF reaches 260k ops/s, outperforming \mpt{} (29k ops/s) by 9$\times$.
\lvmt{} and RainBlock reach 120k and 108k ops/s respectively---roughly half of \mhot{}'s throughput.
Synchronous-flush \mhot{} variants hit 200k ops/s, still 7$\times$ faster than \mpt{}.

As tree size grows, all systems show throughput degradation.
At 500k keys, \mhot{}-AF maintains 135k ops/s while \lvmt{} drops to 90k ops/s and RainBlock to 50k ops/s.
At 1M keys, \mhot{}-AF delivers 104k ops/s versus \mpt{}'s 17k ops/s (6$\times$), \lvmt{}'s 80k ops/s (1.3$\times$), and RainBlock's 38k ops/s (2.7$\times$).
At larger scales, \lvmt{}'s $O(1)$ root update complexity narrows the gap, though \mhot{} retains an advantage through reduced tree traversal depth.

\heading{Real-world trace.}
Under Ethereum mainnet traces, \mhot{}-AF reaches 130k ops/s, outperforming \lvmt{} (72k ops/s) by 1.8$\times$, RainBlock (55k ops/s) by 2.4$\times$, and \mpt{} (20k ops/s) by 6.5$\times$.
The real-world trace exhibits higher key locality than synthetic workloads, benefiting \mhot{}'s cache-friendly compound node layout.
Blake3 variants outperform Keccak variants by 3--5\% due to Blake3's lower computational overhead.
For Ethereum compatibility, Keccak remains the default despite this modest penalty.

\begin{figure}[t]
    \centering
    \includegraphics[width=0.63\columnwidth]{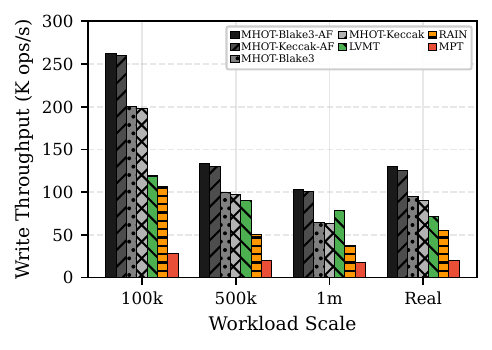}
    \caption{Write throughput. \mhot{}-AF denotes asynchronous flush; \mhot{} without suffix uses synchronous flush. \mhot{} outperforms \mpt{} by 5--9$\times$ across all configurations.} 
    \label{fig:throughput}
        \vspace{-0.1in}
\end{figure}

\subsection{Write Amplification (Figure~\ref{fig:wa})}
\label{sec:eval:wa}

\heading{Synthetic workloads.}
At 100k keys, \mhot{} records average WA of 0.9, greatly lower than \mpt{}'s 2.7, a three times reduction.
\lvmt{} records the lowest WA (0.8) due to its LSM-tree-style append-only storage. RainBlock falls to 1.45.

As tree size grows, \mpt{}'s WA rises from 2.7 at 100k to 4.8 at 1M keys, reflecting the cost of maintaining deep Merkle paths with per-epoch commits.
\mhot{} maintains WA of 0.9--1.6 across scales, a 3--3.7$\times$ reduction over \mpt{}.
\lvmt{} consistently records the lowest WA (0.8--1.0).

\heading{Real-world trace.}
Under Ethereum mainnet traces, \mhot{}-AF records WA of 1.1, versus \mpt{}'s 3.2 (2.9$\times$ reduction), RainBlock's 1.8 (1.6$\times$ reduction), and \lvmt{}'s 1.0.
Synchronous-flush \mhot{} variants show slightly higher WA (1.25) due to more frequent disk commits.

\mhot{}'s batched flush strategy introduces per-epoch variance: most epochs complete with near-zero WA, while flush epochs reach 4--8 depending on accumulated changes.
For aggregate storage efficiency (most relevant to long-running nodes), \mhot{}'s lower WA reduces total I/O over time.
\lvmt{} records 10--15\% lower WA than \mhot{}, but incurs 3--4 orders of magnitude higher verification latency (\S\ref{sec:eval:proof}).

\begin{figure}[t]
    \centering
    \includegraphics[width=0.63\columnwidth]{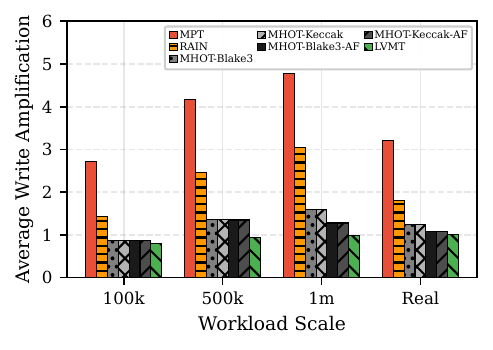}
    \caption{Average write amplification comparison. Lower is better. \lvmt{} achieves the lowest WA via append-only storage. \mhot{} reduces \mpt{}'s WA by 3$\times$ via batched flushing.}
    \label{fig:wa}
        \vspace{-0.2in}
\end{figure}

\subsection{Tree Height Analysis (Figure~\ref{fig:tree_height})}
\label{sec:eval:height}

Tree height directly impacts proof size and verification latency. Each additional level requires more node traversals and more sibling hashes in membership proofs.

\heading{Synthetic workloads.}
Under uniformly distributed keys, \lvmt{} records the lowest tree height of 2 across all scales, reflecting its $2^{16}$ fanout per level. \mhot{} maintains height 5--6, while \mpt{} and RainBlock reach 8--9. \mhot{}'s compound node design yields 35--40\% shallower trees than \mpt{}.

\heading{Real-world trace.}
The Ethereum mainnet trace reveals a limitation of fixed-span architectures. \lvmt{} maintains an average height of 2, but individual branches reach depth 9, approaching \mpt{} and RainBlock's worst-case heights of 10.
Real Ethereum addresses cluster within certain prefixes due to contract factories and sequential account creation, causing fixed 16-bit partitioning to produce unbalanced subtrees.

\mhot{} maintains stable height of 6 under real-world traces.
Its variable-span compound nodes adapt to local key density by packing discriminative bits greedily, absorbing prefix collisions without proportional height increase.
This consistency across workloads matters for blockchain deployments where key distributions vary unpredictably.

\begin{figure}[t]
    \centering
    \includegraphics[width=0.63\columnwidth]{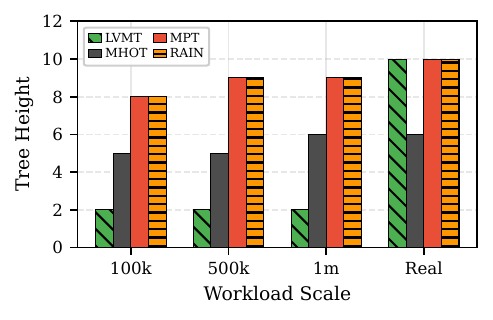}
    \caption{Tree height comparison. \lvmt{} achieves optimal height under synthetic workloads but degrades to match \mpt{} under real traces. \mhot{} maintains consistent height.}
    \label{fig:tree_height}
        \vspace{-0.1in}
\end{figure}

\subsection{Proof Size and Latency}
\label{sec:eval:proof}

We evaluate single-point membership proofs across systems (Figure~\ref{fig:single_proof} and Table~\ref{tab:verify_latency}), then examine \mhot{}'s scalability for multi-point and range proofs (Figure~\ref{fig:proof_scalability}).

\heading{Single-point proofs.}
Figure~\ref{fig:single_proof} compares proof size and prove latency across five system configurations and four tree scales (100K--1M synthetic keys plus Ethereum mainnet trace with 1.6M keys).
We evaluate two \lvmt{} sampling strategies: \emph{random} (best-case, sampling uniformly across keys) and \emph{deepest} (worst-case, targeting keys at maximum tree depth).

\mhot{}-Blake3 achieves the most compact proofs across all scales, ranging from 1,139 bytes at 100K keys to 1,423 bytes under real-world traces.
\lvmt{}-random produces larger proofs (2,227--3,011 bytes) due to KZG commitment overhead, while \lvmt{}-deepest reveals worst-case behavior with proofs reaching 4,411--19,123 bytes under real traces (13$\times$ larger than \mhot{}).
\mpt{} proofs range from 2,304 to 2,867 bytes, twice \mhot{}'s size.
\mhot{}'s proof size advantage stems from its two-layer Merkle architecture: the intra-node Merkle tree requires only $O(\log k)$ sibling hashes per node rather than $O(k)$, where $k=32$ is the maximum fanout.

For prove latency, \lvmt{}-random achieves the lowest values (5.1--7.7\,$\mu$s), while \mhot{}-Blake3 (8.9--11.9\,$\mu$s) outperforms \mhot{}-Keccak (31.5--42.5\,$\mu$s) by 3.5$\times$ due to Blake3's lower computational overhead.
\lvmt{}-deepest reaches 10.6--61.4\,$\mu$s under real traces. \mpt{} (6.2--8.9\,$\mu$s) remains stable.

Table~\ref{tab:verify_latency} presents verification latency.
Hash-based schemes (\mhot{}, \mpt{}) operate in microseconds, while \lvmt{}'s KZG polynomial commitment verification requires milliseconds due to pairing operations.
\mhot{}-Blake3 verifies proofs in 5.3\,$\mu$s, comparable to \mpt{} (10--12\,$\mu$s) while providing smaller proofs.
\lvmt{}'s verification latency of 32--147\,ms represents three to four orders of magnitude overhead compared to hash-based schemes, which is a critical trade-off for applications requiring fast verification such as light clients.

\begin{figure}[t]
    \centering
    \includegraphics[width=0.8\columnwidth]{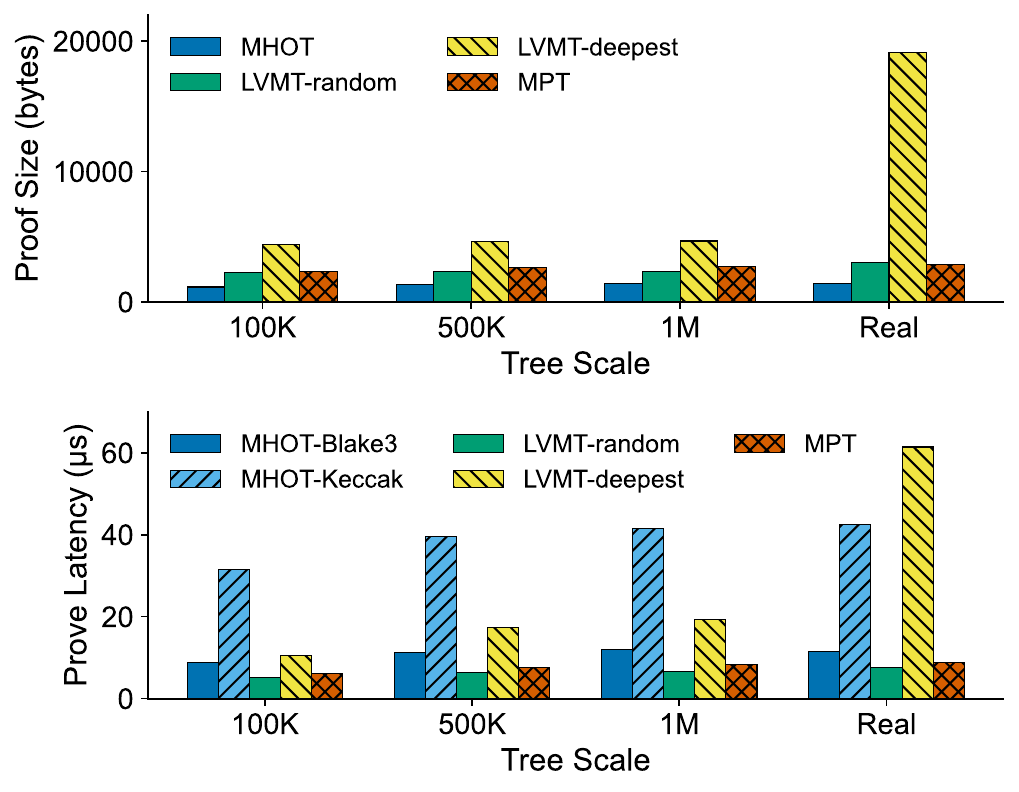}
    \caption{Single-point proof comparison across tree scales. Top panel shows proof size; bottom panel shows prove latency. \lvmt{}-random represents best-case sampling; \lvmt{}-deepest represents worst-case targeting of deep keys. } 
    \label{fig:single_proof}
\end{figure}

\begin{table}[t]
\centering
\caption{Single-point verification latency. Hash-based schemes operate in microseconds; \lvmt{}'s KZG pairing requires milliseconds.}
\label{tab:verify_latency}
\vspace{5pt}
\resizebox{\columnwidth}{!}{%
\begin{tabular}{c|cccc}
\toprule
\textbf{System} & \textbf{100K} & \textbf{500K} & \textbf{1M} & \textbf{Real} \\
\midrule
\mhot{}-Blake3 & 4.4\,$\mu$s & 5.1\,$\mu$s & 5.4\,$\mu$s & 5.3\,$\mu$s \\
\mhot{}-Keccak & 9.9\,$\mu$s & 11.2\,$\mu$s & 12.0\,$\mu$s & 11.9\,$\mu$s \\
\mpt{} & 10.2\,$\mu$s & 11.5\,$\mu$s & 12.3\,$\mu$s & 12.3\,$\mu$s \\
\midrule
\lvmt{} & 33--47\,ms & 32--49\,ms & 32--51\,ms & 38--147\,ms \\
\bottomrule
\end{tabular}%
}
\end{table}

\heading{Multi-point and range proofs.}
Figure~\ref{fig:proof_scalability} shows proof size and latency scaling with batch and range size under Ethereum mainnet traces (1.6M keys).
Multi-point proofs grow from 1.4\,KB (single key) to 92.6\,KB (1000 keys), while range proofs grow from 1.4\,KB to 95.7\,KB.
For latency, multi-point prove time scales from 12\,$\mu$s to 2.3\,ms, while verify time scales from 5.5\,$\mu$s to 4.9\,ms.
Range proofs exhibit similar scaling, with prove latency reaching 2.8\,ms and verify latency reaching 5.1\,ms at 1000 entries.
The rank-based completeness verification (\S\ref{sec:proofs:range}) adds minimal overhead, as rank computation requires only summation over pre-computed leaf counts.

\begin{figure}[t]
    \centering
    \includegraphics[width=\columnwidth]{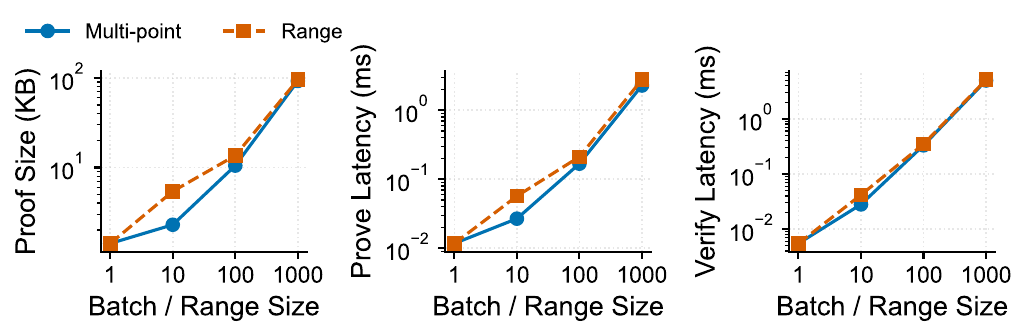}
    \vspace{-0.2in}
    \caption{Multi-point and range proof scalability under Ethereum mainnet traces (Blake3). Proof size and latency scale linearly with batch size. Verify latency dominates at larger batch sizes due to hash recomputation.}
    \label{fig:proof_scalability}
        \vspace{-0.2in}
\end{figure}

\subsection{Nurgle Attack Resistance}
\label{sec:eval:security}

The Nurgle attack~\cite{nurgle} exploits prefix collisions to inflate tree depth for targeted keys, increasing their proof costs.
We evaluate resistance under the threat model where an adversary controls 52 prefix bits (matching the original Nurgle analysis) and commands an entire block's gas budget (300M gas, approximately 15,000 insertions).

\heading{Experimental setup.}
We sample 10,000 random keys as attack targets.
For each target, the attacker generates collision keys matching the target's 52-bit prefix and inserts them until either the target's depth increases or the block's gas budget is exhausted.
For \mpt{}, we employ a round-robin strategy distributing insertions across all targets to maximize attack coverage.
For \lvmt{}, we track the depth distribution of both original and attack keys, since its fixed-depth-at-insertion property prevents existing keys from being pushed deeper.

\heading{Attack results} (Table~\ref{tab:nurgle}). 
\mpt{} proves highly vulnerable: 99.97\% of sampled keys experienced depth increases, with average depth rising from 6.88 to 8.91 (+2.03 levels) and maximum depth increasing from 9 to 12.
In contrast, \mhot{} exhibits strong resistance: zero successful attacks across all 10,000 targets, even after exhausting an entire block's gas budget per target.
\mhot{}'s compound nodes absorb prefix collisions through internal restructuring (Leaf Pushdown, Parent Pull-Up) without propagating depth increases.

\begin{table}[t]
\centering
\small
\caption{Nurgle attack on Ethereum mainnet.}
\label{tab:nurgle}
\vspace{5pt}
\begin{tabular}{c|ccc}
\toprule
\textbf{System} & \textbf{Success Rate} & \textbf{Depth $\Delta$ (avg)} & \textbf{Depth $\Delta$ (max)} \\
\midrule
\mpt{} & 99.97\% & +2.03 & +3 \\
\mhot{} & \textbf{0\%} & \textbf{0} & \textbf{0} \\
\bottomrule
\end{tabular}
\end{table}

\heading{LVMT prefix pollution.}
\lvmt{} presents a different security model: once inserted, a key's depth is fixed and cannot be increased by subsequent insertions.
However, attackers can still pollute prefix regions by inserting keys that occupy slots at shallower levels, forcing future keys into deeper levels.

Figure~\ref{fig:lvmt_pollution} illustrates this effect. Before the attack, 82.3\% of keys reside at level~0 and 17.4\% at level~1 (average depth 0.18). After inserting 15{,}000 attack keys targeting specific prefixes, 99.9\% of attack keys are placed at level~3, far deeper than legitimate keys. Although existing keys remain unaffected, the polluted prefix regions force any future keys sharing these prefixes to be placed at level~3 or deeper. This behavior constitutes a degradation-of-service attack against future users whose addresses collide with the targeted prefixes.

\begin{figure}[t]
    \centering
    \includegraphics[width=0.7\columnwidth]{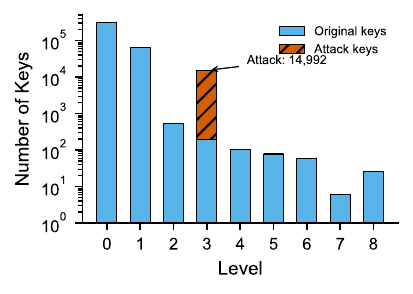}
    \vspace{-0.2in}
    \caption{\lvmt{} key distribution by level before and after Nurgle attack. Attack keys (hatched) concentrate at level 3, polluting targeted prefix regions for future insertions.}
    \label{fig:lvmt_pollution}
\end{figure}

\heading{Key experimental findings.}
\label{sec:eval:findings}
Our experiments yield four key findings.
\mhot{} achieves 5--9$\times$ higher write throughput than \mpt{}, reaching up to 260k ops/s at 100k keys and 100k ops/s at 1M keys (\textbf{Q1}).
This gain is partly due to reduced write amplification: \mhot{} achieves average WA of 0.9--1.6, a 3--4$\times$ reduction compared to \mpt{} (\textbf{Q2}).
\mhot{} also maintains a stable tree height of 5--6 levels across all workloads, 35--40\% shallower than \mpt{}, resulting in 50\% smaller membership proofs (1.1--1.4\,KB vs. 2.3--2.9\,KB) with comparable verification latency (\textbf{Q3}).
Finally, under the Nurgle threat model with a 15,000-insertion budget, none of the 10,000 sampled keys experienced depth increases (\textbf{Q4}).

\section{Discussion}
\label{sec:discussion}

We discuss two aspects as below. 

\subsection{Scope and Complementary Mitigations}
\label{sec:scope}

\heading{Our scope.} Our evaluation targets the state commitment bottleneck in modern clients. Integrating \mhot{} into full execution pipelines would validate its end-to-end impact on block processing.
\mhot{} operates at the ADS layer beneath Ethereum's account model~\cite{ethereum}, treating account fields as opaque key-value pairs. It therefore remains compatible with state-clearing semantics~\cite{eip161}, proof interfaces~\cite{eip1186}, and gas accounting standards~\cite{eip2929}. Deployment can follow an incremental migration path where new commits use \mhot{} while historical proofs remain verifiable against archived \mpt{} roots.

Reduced commitment latency may also benefit block validation throughput, though quantifying this effect requires end-to-end pipeline evaluation.

\heading{Complementary mitigation strategies.}
In-memory layering caches reused upper trie levels to reduce authenticated-state cost without changing the local tree. LMPTs~\cite{lmpts} keep recent-update tries in memory and the snapshot trie on disk; our RainBlock-style baseline keeps the upper six levels in memory and loads deeper nodes from RocksDB. This helps the normal case where updates reuse upper ancestors, but it does not close the Nurgle path. Adversarial keys share long prefixes and grow paths below the cached levels (a single RTX3080 GPU can manipulate the first 15 \mpt{} layers by colliding 13 nibbles~\cite{nurgle}), and the tree still routes by fixed prefixes regardless of how deep the cache extends.

Systemic mitigations reorganize state above the local tree. Chainspace~\cite{chainspace} splits state across shards and certifies shard-local commitments through quorum signatures, while RainBlock~\cite{rainblock} distributes state in DSM-TREE shards and offloads disk access to storage nodes. These designs expose parallelism above the tree at the cost of data placement, network coordination, quorum certification, and cross-shard access, and compose with \mhot{}, since each shard or checkpoint still requires a local authenticated structure that remains exposed to depth inflation when prefix-based. \mhot{} fills that local structural role inside layered, partitioned, or checkpointed deployments.

\subsection{Future Directions}
\label{sec:future}

\heading{Hardware acceleration.} Increasing the span to $k=64$~\cite{hot} exploits AVX-512 instructions on modern processors, further reducing tree height. For spans of $k \ge 256$, scalar implementations suffice because persistent storage I/O dominates latency. Alternatively, algebraic commitments such as vector commitments in AMT~\cite{lvmt} trade setup transparency for reduced verification overhead. \mhot{}'s deferred hashing exposes parallelism well-suited to GPU acceleration, where nodes at the same height level are mutually independent and map efficiently onto GPU SIMD units~\cite{deng2024gpu}.

\heading{Storage optimization.} Integrating HOT's leaf-optimized node layout would substantially reduce node sizes through dense-region collapsing and variable-length delta encoding~\cite{hot}, though deterministic Merkle hash computation must remain tractable.

\heading{State pruning.} Content-addressable, copy-on-write storage creates new node versions on every modification. Version-prefixed database keys allow efficient range deletion of nodes older than a retention threshold, and the underlying LSM-tree engine's compaction naturally discards tombstoned entries. A production deployment would benefit from a configurable pruning policy that balances historical state availability against storage growth.

\heading{Concurrency.} HOT's copy-on-write semantics with wait-free readers~\cite{hot} provide a foundation for concurrent access. \mhot{}'s immutable nodes already support concurrent reads without synchronization. However, coordinating parallel writers remains an open problem.

\heading{Structural attack analysis.}
\mhot{} removes Nurgle's prefix-collision attack through discriminative-bit indexing, but its structure may open new attack vectors. An adversary could craft keys that concentrate discriminative-bit conflicts within certain subtrees, forcing repeated node splits such as leaf pushdown, parent pull-up, and intermediate node creation. Following Nurgle's analytical approach, a systematic study would estimate the effort to force such splits, measure the resulting path-length increase per malicious insertion, and evaluate whether these effects produce economic imbalance under the gas pricing model~\cite{eip2929}.
We leave tight bounds on adversarial path inflation to future work.

\section{Related Work}
\label{sec:related}

\heading{Protocol-level updates.}
Ethereum’s roadmap explores replacing \mpt{} with alternative state commitment structures.
Verkle Trees~\cite{verkle}, proposed in EIP-6800~\cite{eip6800}, adopt vector commitments to achieve constant-size proofs per tree level, enabling stateless client verification. This design relies on a trusted setup; while multi-party computation ceremonies distribute trust, compromise remains possible.
As an alternative, EIP-7864~\cite{eip7864} proposes a binary Merkle tree tailored for SNARK-based proof generation, partly motivated by post-quantum concerns of pairing-based schemes~\cite{postquantum-verkle}. Both approaches require hard forks and incur nontrivial deployment costs, either through additional trust assumptions or substantial proof generation overhead.

In contrast, \mhot{} demonstrates that performance and robustness gains remain achievable within the existing hash-based commitment model. More importantly, \mhot{} preserves the simplicity of Merkle authentication without introducing algebraic verification overhead for deployment.

\heading{Vector commitments.}
VC literature addresses a different layer, asking how to commit to a value set and open positions with short proofs rather than how state is organized, indexed, or traversed; the two layers compose. When paired with a trie, as in Verkle and \lvmt{} (\S\ref{sec:span-proof}), a VC decouples proof size from fanout cryptographically, while \mhot{} achieves the same decoupling structurally via discriminative-bit indexing and hierarchical proofs, a different design point, not a competing one. Pure VC constructions without a tree, such as Aardvark~\cite{aardvark} (bucketed dictionary), KVaC~\cite{kvac} (flat key-value commitment), and EDRAX~\cite{edrax} (indexed authenticated array), target different objectives and do not map onto the axes of Table~\ref{tab:comparison}. Proof-aggregation techniques (Pointproofs~\cite{pointproofs}, Hyperproofs~\cite{hyperproofs}, Cauchyproofs~\cite{cauchyproofs}, aSVC~\cite{asvc}) are orthogonal to the base structure; foundational results~\cite{catalano-fiore13,campanelli2020,cfgg22} establish the VC primitive, incremental aggregation, and impossibility bounds. The lack of native range proofs in \lvmt{} and Verkle (Table~\ref{tab:comparison}) reflects engineering status, not a fundamental VC limitation; interval certification needs extra construction that hash-based Merkle proofs avoid.

\heading{State storage optimization.}
Beyond modifications to the authenticated tree structure itself, prior work has explored workload-aware optimizations for blockchain state storage. Adaptive tree restructuring~\cite{kuznetsov2024} dynamically promotes frequently accessed nodes toward the root, reducing average access latency for hot state. Hot–cold data separation schemes~\cite{feng2025} migrate infrequently accessed state to lower-cost storage tiers while preserving fast access to active accounts. These approaches primarily optimize for access frequency and locality, rather than the structural properties of the key space.

In contrast, \mhot{}’s height optimization targets worst-case structural depth induced by key relationships, independent of workload skew. The two strategies are orthogonal and can be naturally composed.

\heading{Storage architecture.}
Blockchain-aware storage engines mitigate I/O amplification by exploiting blockchain-specific access patterns, primarily through two complementary strategies.
The first leverages key structure and layout. ChainKV~\cite{chainkv} separates state from non-state data using Prefix-\mpt{} to improve key locality, while Block-LSM~\cite{blocklsm} prefixes keys with block numbers to cluster same-block writes and reduce compaction overhead. The second strategy amortizes commitment cost over time. LETUS~\cite{letus} employs log-structured delta encoding across blocks, and COLE~\cite{cole} applies learned indexes to optimize read-heavy workloads.

Modern production clients such as Erigon~\cite{erigon} further adopt flat database designs that decouple state access from commitment computation, effectively treating the authenticated data structure as a dedicated commitment engine.
\mhot{} follows this decoupling principle and focuses optimization squarely on the commitment bottleneck itself.

\heading{Acceleration techniques.}
Hardware acceleration complements algorithmic improvements. Deng et al.~\cite{deng2024gpu} parallelize \mpt{} hash computation on GPUs via PhaseNU and LockNU algorithms, addressing node-splitting conflicts during concurrent updates. This achieves substantial throughput gains for commitment-intensive workloads. Combining GPU parallelism with height-optimized structures remains unexplored.

\mhot{}'s deferred hashing exposes parallelism suited for such acceleration: nodes at the same height level share no data dependencies during hash computation.

\heading{Nurgle mitigation.}
The attack~\cite{nurgle} exploits the predictable structure of \mpt{} to inflate tree depth via adversarial key selection. Current mitigation efforts focus on economic disincentives or data pruning rather than structural defenses.
EIP-4444~\cite{eip4444} enables historical data pruning but leaves the current state structure unchanged. EIP-4762~\cite{eip4762} proposes gas repricing for witness costs but remains in draft, facing determinism challenges. Verkle migration~\cite{verkle,eip6800} would alter the attack surface but requires ongoing protocol changes.

\mhot{} provides an immediate structural defense.
Even when an adversary controls an entire block's gas budget, \mhot{} achieves zero successful depth increases, compared to 99.97\% attack success rate against \mpt{}.
The compound node design absorbs prefix collisions through internal restructuring without propagating depth increases to existing keys.

\begin{algorithm}[H]
\caption{Compact Multiproof Generation~\cite{merkle-multiproof}}
\label{alg:multiproof-gen}
\linespread{0.98}\footnotesize
\begin{algorithmic}[1]
\Require leaves $L$, indices $I$ to prove
\Ensure Compact multiproof $\pi$
\State $\text{depth} \gets \lceil \log_2 |L| \rceil$
\State $\text{known} \gets \{(\text{depth}, i) : i \in I\}$
\State $\text{proof\_hashes} \gets \langle\rangle$
\For{$\ell \gets \text{depth}-1$ \textbf{downto} $0$}
    \For{$i \gets 0$ \textbf{to} $2^\ell - 1$}
        \State $\text{left} \gets (\ell+1, 2i) \in \text{known}$
        \State $\text{right} \gets (\ell+1, 2i+1) \in \text{known}$
        \If{left \textbf{and} right}
            \State $\text{known}.\textsc{Insert}((\ell, i))$
        \ElsIf{left \textbf{or} right}
            \State $s \gets 2i+1$ \textbf{if} left \textbf{else} $2i$
            \State $\text{proof\_hashes}.\textsc{Push}(\textsc{Hash}(\ell+1, s))$
            \State $\text{known}.\textsc{Insert}((\ell, i))$
        \EndIf
    \EndFor
\EndFor
\State \Return $(I, \text{proof\_hashes}, \text{depth})$
\end{algorithmic}
\end{algorithm}

\section{Conclusion}
\label{sec:conclusion}
We presented \mhot{}, a height-optimal authenticated data structure for blockchain state commitment. By adapting height-optimized tries to persistent storage, \mhot{} achieves substantially higher throughput, lower write amplification, and smaller proofs than \mpt{}, without relying on trusted setup or specialized cryptography. \mhot{} structurally mitigates Nurgle attacks, maintaining zero successful depth increases even under worst-case adversarial conditions. Our results show that careful data-structure design can fundamentally improve the scalability and robustness of blockchain state commitment.

\appendix

\bibliographystyle{plainurl}
\bibliography{bib}

@article{hot,
  title={Height Optimized Tries},
  author={Robert Binna and Eva Zangerle and Martin Pichl and G{\"u}nther Specht and Viktor Leis},
  journal={ACM Transactions on Database Systems (TODS)},
  year={2022},
  volume={47},
  pages={1 - 46}
}

@inproceedings{nurgle,
  title={Nurgle: Exacerbating Resource Consumption in Blockchain State Storage via MPT Manipulation},
  author={Zheyuan He and Zihao Li and Ao Qiao and Xiapu Luo and Xiaosong Zhang and Ting Chen and Shuwei Song and Dijun Liu and Weina Niu},
  booktitle={IEEE Symposium on Security and Privacy (S\&P)},
  year={2024},
  pages={2180-2197}
}

@article{lvmt,
  author = {Li, Chenxing and Beillahi, Sidi Mohamed and Yang, Guang and Wu, Ming and Xu, Wei and Long, Fan},
  title = {LVMT: An Efficient Authenticated Storage for Blockchain},
  year = {2024},
  issue_date = {August 2024},
  publisher = {Association for Computing Machinery},
  address = {New York, NY, USA},
  volume = {20},
  number = {3},
  issn = {1553-3077},
  doi = {10.1145/3664818},
  journal = {ACM Transactions on Storage},
  month = jun,
  articleno = {17},
  numpages = {34}
}

@inproceedings{verkle,
  title={Verkle Trees},
  author={John Kuszmaul},
  year={2019},
  note={Available at: \url{https://math.mit.edu/research/highschool/primes/materials/2018/Kuszmaul.pdf}}
}

@misc{ethereum,
  title={Ethereum: A Secure Decentralised Generalised Transaction Ledger},
  author={Gavin Wood},
  year={2014},
  howpublished={Ethereum Yellow Paper},
  note={Available at: \url{https://ethereum.github.io/yellowpaper/paper.pdf}}
}

@article{chainkv,
  title={ChainKV: A Semantics-Aware Key-Value Store for Ethereum System},
  author={Zehao Chen and Bingzhe Li and Xiaojun Cai and Zhiping Jia and Lei Ju and Zili Shao and Zhaoyan Shen},
  journal={Proceedings of the ACM on Management of Data (SIGMOD)},
  year={2023},
}

@inproceedings{rainblock,
  title={RainBlock: Faster Transaction Processing in Public Blockchains},
  author={Soujanya Ponnapalli and Aashaka Shah and Amy Tai and Souvik Banerjee and Vijay Chidambaram and Dahlia Malkhi and Michael Yung Chung Wei},
  booktitle={USENIX Annual Technical Conference (USENIX ATC)},
  year={2020}
}

@inproceedings{lmpts,
  title={LMPTs: Eliminating Storage Bottlenecks for Processing Blockchain Transactions},
  author={Jemin Andrew Choi and Sidi Mohamed Beillahi and Peilun Li and Andreas G. Veneris and Fan Long},
  booktitle={IEEE International Conference on Blockchain and Cryptocurrency (ICBC)},
  year={2022},
  pages={1-9}
}

@misc{erigon,
  title = {Erigon: Ethereum Implementation on the Efficiency Frontier},
  author = {{Erigon Team}},
  howpublished = {GitHub Repository},
  year = {2024},
  note = {Available at: \url{https://github.com/erigontech/erigon}}
}

@misc{reth,
  title = {Reth: Modular, Contributor-Friendly and Blazing-Fast Implementation of the Ethereum Protocol in Rust},
  author = {{Paradigm}},
  howpublished = {GitHub Repository},
  year = {2024},
  note = {Available at: \url{https://github.com/paradigmxyz/reth}}
}

@article{blocklsm,
  title={A Semantic-Integrated LSM-Tree-Based Key-Value Storage Engine for Blockchain Systems},
  author={Qian Wei and Zehao Chen and Xiaowei Chen and Yuhao Zhang and Xiaojun Cai and Zhiping Jia and Zhaoyan Shen and Yi Wang and Zili Shao and Bingzhe Li},
  journal={IEEE Transactions on Computer-Aided Design of Integrated Circuits and Systems (TCAD)},
  year={2024},
}

@inproceedings{letus,
  title={LETUS: A Log-Structured Efficient Trusted Universal BlockChain Storage},
  author={Shikun Tian and Zhonghao Lu and Haizhen Zhuo and Xiaojing Tang and Peiyi Hong and Shenglong Chen and Dayi Yang and Ying Yan and Zhiyong Jiang and Hui Zhang and Guofei Jiang},
  booktitle={Companion of the 2024 International Conference on Management of Data (SIGMOD)},
  year={2024}
}

@inproceedings{cole,
  title={COLE: A Column-based Learned Storage for Blockchain Systems},
  author={Ce Zhang and Cheng Xu and Haibo Hu and Jianliang Xu},
  booktitle={USENIX Conference on File and Storage Technologies (FAST)},
  year={2023}
}

@inproceedings{kzg,
  title={Constant-Size Commitments to Polynomials and Their Applications},
  author={Aniket Kate and Gregory M. Zaverucha and Ian Goldberg},
  booktitle={International Conference on the Theory and Application of Cryptology and Information Security (ASIACRYPT)},
  year={2010}
}

@inproceedings{amt,
  title={Towards Scalable Threshold Cryptosystems},
  author={Alin Tomescu and Robert Chen and Yiming Zheng and Ittai Abraham and Benny Pinkas and Guy Golan-Gueta and Srinivas Devadas},
  booktitle={IEEE Symposium on Security and Privacy (S\&P)},
  year={2020},
  pages={877-893}
}

@article{ads,
  title={A General Model for Authenticated Data Structures},
  author={Charles U. Martel and Glen Nuckolls and Premkumar T. Devanbu and Michael Gertz and April Kwong and Stuart G. Stubblebine},
  journal={Algorithmica},
  year={2004},
  volume={39},
  pages={21-41}
}

@misc{eip150,
  title = {{EIP}-150: Gas Cost Changes for IO-Heavy Operations},
  author = {Vitalik Buterin},
  year = {2016},
  month = oct,
  howpublished = {Ethereum Improvement Proposals},
  note = {\url{https://eips.ethereum.org/EIPS/eip-150}}
}

@misc{eip161,
  title = {{EIP}-161: State Trie Clearing (Invariant-Preserving Alternative)},
  author = {Gavin Wood},
  year = {2016},
  month = oct,
  howpublished = {Ethereum Improvement Proposals},
  note = {\url{https://eips.ethereum.org/EIPS/eip-161}}
}

@misc{eip1186,
  title = {{EIP}-1186: RPC-Method to Get Merkle Proofs - eth\_getProof},
  author = {Simon Jentzsch and Christoph Jentzsch},
  year = {2018},
  month = jun,
  howpublished = {Ethereum Improvement Proposals},
  note = {\url{https://eips.ethereum.org/EIPS/eip-1186}}
}

@misc{eip2929,
  title = {{EIP}-2929: Gas Cost Increases for State Access Opcodes},
  author = {Vitalik Buterin and Martin Swende},
  year = {2020},
  month = sep,
  howpublished = {Ethereum Improvement Proposals},
  note = {\url{https://eips.ethereum.org/EIPS/eip-2929}}
}

@misc{eip3102,
  title = {{EIP}-3102: Binary Trie Structure},
  author = {Guillaume Ballet and Vitalik Buterin},
  year = {2020},
  month = sep,
  howpublished = {Ethereum Improvement Proposals},
  note = {\url{https://eips.ethereum.org/EIPS/eip-3102}}
}

@misc{eip4444,
  title = {{EIP}-4444: Bound Historical Data in Execution Clients},
  author = {George Kadianakis and lightclient and Alex Stokes},
  year = {2021},
  month = nov,
  howpublished = {Ethereum Improvement Proposals},
  note = {\url{https://eips.ethereum.org/EIPS/eip-4444}}
}

@misc{eip4762,
  title = {{EIP}-4762: Statelessness Gas Cost Changes},
  author = {Guillaume Ballet and Vitalik Buterin and Dankrad Feist and Ignacio Hagopian and Tanishq Jasoria and Gajinder Singh},
  year = {2022},
  month = feb,
  howpublished = {Ethereum Improvement Proposals},
  note = {\url{https://eips.ethereum.org/EIPS/eip-4762}}
}

@misc{eip6800,
  title = {{EIP}-6800: Ethereum State Using a Unified Verkle Tree},
  author = {Vitalik Buterin and Dankrad Feist and Kevaundray Wedderburn and Guillaume Ballet and Piper Merriam and Gottfried Herold and Ignacio Hagopian and Tanishq Jasoria and Gajinder Singh and Danno Ferrin},
  year = {2023},
  month = mar,
  howpublished = {Ethereum Improvement Proposals},
  note = {\url{https://eips.ethereum.org/EIPS/eip-6800}}
}

@misc{eip7864,
  title = {{EIP}-7864: Ethereum State Using a Unified Binary Tree},
  author = {Vitalik Buterin and Guillaume Ballet and Dankrad Feist and Ignacio Hagopian and Kevaundray Wedderburn and Tanishq Jasoria and Gajinder Singh and Danno Ferrin and Piper Merriam and Gottfried Herold},
  year = {2025},
  month = jan,
  howpublished = {Ethereum Improvement Proposals},
  note = {\url{https://eips.ethereum.org/EIPS/eip-7864}}
}

@inproceedings{occda,
  title={Utilizing Parallelism in Smart Contracts on Decentralized Blockchains by Taming Application-Inherent Conflicts},
  author={Garamv{\"o}lgyi, P{\'e}ter and Liu, Yuxi and Zhou, Dong and Long, Fan and Wu, Ming},
  booktitle={International Conference on Software Engineering (ICSE)},
  pages={2315--2326},
  year={2022}
}

@inproceedings{blockstm,
  title={Block-STM: Scaling Blockchain Execution by Turning Ordering Curse to a Performance Blessing},
  author={Gelashvili, Rati and Spiegelman, Alexander and Xiang, Zhuolun and Danezis, George and Li, Zekun and Malkhi, Dahlia and Xia, Yu and Zhou, Runtian},
  booktitle={ACM SIGPLAN Annual Symposium on Principles and Practice of Parallel Programming (PPoPP)},
  pages={232--244},
  year={2023}
}

@article{schain,
  title={SChain: Scalable Concurrency over Flexible Permissioned Blockchain},
  author={Xiaodong Qi and Zhihao Chen and Haizhen Zhuo and Quanqing Xu and Chengyu Zhu and Zhao Zhang and Cheqing Jin and Aoying Zhou and Ying Yan and Hui Zhang},
  journal={IEEE International Conference on Data Engineering (ICDE)},
  year={2023},
  pages={1901-1913},
}

@inproceedings{parallelevm,
  title={ParallelEVM: Operation-Level Concurrent Transaction Execution for EVM-Compatible Blockchains},
  author={Haoran Lin and Hang Feng and Yajin Zhou and Lei Wu},
  booktitle={Proceedings of the European Conference on Computer Systems (EuroSys)},
  year={2025}
}

@inproceedings{forerunner,
  title     = {Forerunner: Constraint-Based Speculative Transaction Execution for Ethereum},
  author    = {Chen, Yang and Guo, Zhongxin and Li, Runhuai and Chen, Shuo and Zhou, Lidong and Zhou, Yajin and Zhang, Xian},
  booktitle = {ACM SIGOPS Symposium on Operating Systems Principles (SOSP)},
  pages     = {570--587},
  year      = {2021}
}

@article{seer,
  title     = {Seer: Accelerating Blockchain Transaction Execution by Fine-Grained Branch Prediction},
  author    = {Zhang, Shijie and Cheng, Ru and Liu, Xinpeng and Xiao, Jiang and Jin, Hai and Li, Bo},
  journal   = {Proceedings of the VLDB Endowment (VLDB)},
  volume    = {18},
  number    = {3},
  pages     = {822--835},
  year      = {2024}
}

@inproceedings{mtpu,
  title   = {An Algorithm and Architecture Co-design for Accelerating Smart Contracts in Blockchain},
  author  = {Rui Pan and Chubo Liu and Guoqing Xiao and Mingxing Duan and Keqin Li and Kenli Li},
  booktitle = {Annual International Symposium on Computer Architecture (ISCA)},
  year    = {2023}
}

@article{dtvm,
  title   = {{DTVM}: Revolutionizing Smart Contract Execution with Determinism and Compatibility},
  author  = {Wei Zhou and Changzheng Wei and Ying Yan and Wei Tang and others},
  journal = {arXiv preprint arXiv:2504.16552},
  year    = {2025},
  note    = {Available at: \url{https://arxiv.org/abs/2504.16552}}
}

@inproceedings{superinstruction,
  title={Synthesizing Efficient Super-Instruction Sets for Ethereum Virtual Machine},
  author={Xiaowen Hu and David Zhao and Bernhard Scholz},
  booktitle={ACM SIGPLAN International Workshop on Virtual Machines and Intermediate Languages (VMIL)},
  year={2024}
}

@article{deng2024gpu,
  title={Accelerating Merkle Patricia Trie with GPU},
  author={Yangshen Deng and Muxi Yan and Bo Tang},
  journal={Proceedings of the VLDB Endowment (VLDB)},
  year={2024},
  volume={17},
  pages={1856-1869}
}

@misc{geth-parallel-node-fetching,
  author       = {{Geth Team}},
  title        = {Parallel Intermediate Node Fetching (for a single trie)},
  howpublished = {Go Ethereum Issue \#28266},
  year         = {2023},
  note         = {Available at: \url{https://github.com/ethereum/go-ethereum/issues/28266}. Accessed: 2026-05-18}
}

@article{splitDB,
  title={SplitDB: Closing the Performance Gap for LSM-Tree-Based Key-Value Stores},
  author={Miao Cai and Xuzhen Jiang and Junru Shen and Baoliu Ye},
  journal={IEEE Transactions on Computers (TC)},
  year={2024},
  volume={73},
  pages={206-220}
}

@article{solsDB,
  title={SolsDB: Solve the Ethereum's Bottleneck Caused by Storage Engine},
  author={Cuihua Yang and Fan Yang and Quanqing Xu and Yongquan Zhang and Junqing Liang},
  journal={Future Generation Computer Systems (FGCS)},
  year={2024},
  volume={160},
  pages={295-304}
}

@article{kuznetsov2024,
  title={Adaptive Restructuring of Merkle and Verkle Trees for Enhanced Blockchain Scalability},
  author={Oleksandr Kuznetsov and Dzianis Kanonik and Alex Rusnak and Anton Yezhov and Oleksandr Domin},
  journal={Internet of Things},
  year={2024},
  volume={27},
  pages={101315}
}

@article{feng2025,
  title={An Efficient Hot/Cold Data Separation Scheme for Storage Optimization in Consortium Blockchain Full Nodes},
  author={Libo Feng and Xian Deng},
  journal={Cluster Computing},
  year={2025},
  volume={28}
}

@inproceedings{postquantum-verkle,
  title={Functional Commitments for All Functions, with Transparent Setup},
  author={Leo de Castro and Chris Peikert},
  booktitle={Annual International Conference on the Theory and Applications of Cryptographic Techniques (EUROCRYPT)},
  year={2023}
}

@article{lsmtree,
  title={The Log-Structured Merge-Tree (LSM-tree)},
  author={O'Neil, Patrick and Cheng, Edward and Gawlick, Dieter and O'Neil, Elizabeth},
  journal={Acta Informatica},
  volume={33},
  number={4},
  pages={351--385},
  year={1996}
}

@article{smhp,
  title={Partitioning of Trees for Minimizing Height and Cardinality},
  author={Andr{\'a}s Kov{\'a}cs and Tam{\'a}s Kis},
  journal={Information Processing Letters},
  year={2004},
  volume={89},
  pages={181-185}
}

@inproceedings{rocksdb,
  author    = {Siying Dong and Andrew Kryder and Yanqin Jin and Lin Peng and Kanchan Mehra and Jeremy Yakdus and Wei-Nee Chen and Abhishek Sharma and Youngjin Kwon and Gary J. Katz},
  title     = {{RocksDB}: Evolution of Development, Optimization and Uses of LSM-based Storage},
  booktitle = {Proceedings of the 8th Biennial Conference on Innovative Data Systems Research (CIDR)},
  year      = {2017}
}

@misc{merkle-multiproof,
  title = {Merkle Multi-Proofs},
  author = {Remco Bloemen},
  year = {2025},
  howpublished = {Technical Note},
  note = {Available at: \url{https://xn--2-umb.com/25/merkle-multi-proof/}}
}

@article{shoup2004sequences,
    title={Sequences of games: a tool for taming complexity in security proofs},
    author={Shoup, Victor},
    journal={IACR Cryptol. ePrint Arch.},
    volume={2004},
    pages={332},
    year={2004}
}

@misc{binary,
  title = {{EIP}-7864: Ethereum State Using a Unified Binary Tree},
  author = {Vitalik Buterin and Guillaume Ballet and Dankrad Feist and Ignacio Hagopian and Kevaundray Wedderburn and Tanishq Jasoria and Gajinder Singh and Danno Ferrin and Piper Merriam and Gottfried Herold},
  year = {2025},
  month = jan,
  howpublished = {Ethereum Improvement Proposals},
  note = {\url{https://eips-wg.github.io/EIPs/7864/}}
}

@inproceedings{merkle,
  title={A Digital Signature Based on a Conventional Encryption Function},
  author={Ralph C. Merkle},
  booktitle={Annual International Cryptology Conference},
  year={1987},
  url={https://api.semanticscholar.org/CorpusID:28484604}
}

@misc{devp2p,
	title = {devp2p/caps/snap.md at master · ethereum/devp2p},
	url = {https://github.com/ethereum/devp2p/blob/master/caps/snap.md},
	urldate = {2026-02-05},
}

@inproceedings{chainspace,
  title = {Chainspace: A Sharded Smart Contracts Platform},
  author = {Mustafa Al-Bassam and Alberto Sonnino and Shehar Bano and Dave Hrycyszyn and George Danezis},
  booktitle = {Network and Distributed System Security Symposium (NDSS)},
  year = {2018},
  doi = {10.14722/ndss.2018.23241}
}

@inproceedings{rangeproofimportant,
  author       = {Enrique Fynn and
                  Ethan Buchman and
                  Zarko Milosevic and
                  Robert Soul{\'{e}} and
                  Fernando Pedone},
  editor       = {Eshcar Hillel and
                  Roberto Palmieri and
                  Etienne Rivi{\`{e}}re},
  title        = {Robust and Fast Blockchain State Synchronization},
  booktitle    = {26th International Conference on Principles of Distributed Systems,
                  {OPODIS} 2022, Brussels, Belgium, December 13-15, 2022},
  series       = {LIPIcs},
  volume       = {253},
  pages        = {8:1--8:22},
  publisher    = {Schloss Dagstuhl - Leibniz-Zentrum f{\"{u}}r Informatik},
  year         = {2022},
  url          = {https://doi.org/10.4230/LIPIcs.OPODIS.2022.8},
  doi          = {10.4230/LIPICS.OPODIS.2022.8},
  bibsource    = {dblp computer science bibliography, https://dblp.org}
}

@inproceedings{catalano-fiore13,
  title={Vector Commitments and Their Applications},
  author={Dario Catalano and Dario Fiore},
  booktitle={International Conference on Practice and Theory in Public-Key Cryptography (PKC)},
  series={LNCS},
  volume={7778},
  pages={55--72},
  publisher={Springer},
  year={2013},
  doi={10.1007/978-3-642-36362-7_5}
}

@inproceedings{campanelli2020,
  title={Incrementally Aggregatable Vector Commitments and Applications to Verifiable Decentralized Storage},
  author={Matteo Campanelli and Dario Fiore and Nicola Greco and Dimitris Kolonelos and Luca Nizzardo},
  booktitle={International Conference on the Theory and Application of Cryptology and Information Security (ASIACRYPT)},
  series={LNCS},
  volume={12492},
  pages={3--35},
  publisher={Springer},
  year={2020},
  doi={10.1007/978-3-030-64834-3_1}
}

@inproceedings{cfgg22,
  title={On the Impossibility of Algebraic Vector Commitments in Pairing-Free Groups},
  author={Dario Catalano and Dario Fiore and Rosario Gennaro and Emanuele Giunta},
  booktitle={Theory of Cryptography Conference (TCC)},
  series={LNCS},
  volume={13748},
  pages={274--299},
  publisher={Springer},
  year={2022},
  doi={10.1007/978-3-031-22365-5_10}
}

@inproceedings{aardvark,
  title={Aardvark: An Asynchronous Authenticated Dictionary with Applications to Account-based Cryptocurrencies},
  author={Derek Leung and Yossi Gilad and Sergey Gorbunov and Leonid Reyzin and Nickolai Zeldovich},
  booktitle={31st USENIX Security Symposium (USENIX Security 22)},
  pages={4237--4254},
  publisher={USENIX Association},
  year={2022}
}

@inproceedings{kvac,
  title={KVaC: Key-Value Commitments for Blockchains and Beyond},
  author={Shashank Agrawal and Srinivasan Raghuraman},
  booktitle={International Conference on the Theory and Application of Cryptology and Information Security (ASIACRYPT)},
  series={LNCS},
  volume={12493},
  pages={839--869},
  publisher={Springer},
  year={2020},
  doi={10.1007/978-3-030-64840-4_28}
}

@misc{edrax,
  title={EDRAX: A Cryptocurrency with Stateless Transaction Validation},
  author={Alexander Chepurnoy and Charalampos Papamanthou and Shravan Srinivasan and Yupeng Zhang},
  year={2018},
  howpublished={Cryptology ePrint Archive, Report 2018/968},
  note={\url{https://eprint.iacr.org/2018/968}}
}

@inproceedings{pointproofs,
  title={Pointproofs: Aggregating Proofs for Multiple Vector Commitments},
  author={Sergey Gorbunov and Leonid Reyzin and Hoeteck Wee and Zhenfei Zhang},
  booktitle={ACM SIGSAC Conference on Computer and Communications Security (CCS)},
  pages={2007--2023},
  publisher={ACM},
  year={2020},
  doi={10.1145/3372297.3417244}
}

@inproceedings{hyperproofs,
  title={Hyperproofs: Aggregating and Maintaining Proofs in Vector Commitments},
  author={Shravan Srinivasan and Alexander Chepurnoy and Charalampos Papamanthou and Alin Tomescu and Yupeng Zhang},
  booktitle={USENIX Security Symposium},
  pages={3001--3018},
  publisher={USENIX},
  year={2022}
}

@inproceedings{cauchyproofs,
  title={Cauchyproofs: Batch-Updatable Vector Commitment with Easy Aggregation and Application to Stateless Blockchains},
  author={Zhongtang Luo and Yanxue Jia and Alejandra Victoria Ospina Gracia and Aniket Kate},
  booktitle={IEEE Symposium on Security and Privacy (S\&P)},
  publisher={IEEE},
  year={2025},
  doi={10.1109/SP61157.2025.00247}
}

@inproceedings{asvc,
  title={Aggregatable Subvector Commitments for Stateless Cryptocurrencies},
  author={Alin Tomescu and Ittai Abraham and Vitalik Buterin and Justin Drake and Dankrad Feist and Dmitry Khovratovich},
  booktitle={Security and Cryptography for Networks (SCN)},
  series={LNCS},
  volume={12238},
  pages={45--64},
  publisher={Springer},
  year={2020},
  doi={10.1007/978-3-030-57990-6_3}
}

\section{Notations (Table~\ref{tab:notation})}

\begin{table}[t]
\centering
\caption{Summary of notation.}
\vspace{5pt}
\label{tab:notation}
\begin{tabular}{c|l}
\toprule
\multicolumn{1}{c}{\textbf{Symbol}} & \multicolumn{1}{c}{\textbf{Description}} \\
\midrule
$s$, $n$, $k$, $h$ & Span, entry count, fanout ($k{=}32$), tree height \\
$K$, $V$, $T$ & Key, value, trie \\
$H$, $\pi$, $\parallel$ & Hash function, proof, concatenation \\
\midrule
$N$ & Compound node \\
$M$, $S$, $L$ & Extraction masks, sparse keys, leaf counts \\
$\eta$, $v$, $c_i$ & Node height, version, $i$-th child \\
\midrule
$\mathrm{CMR}(N)$ & Children Merkle root \\
$J$, $t$ & Child index set, $t = |J|$ \\
$\mathrm{NPE}(N, J)$ & Node proof entry for children $J$ \\
$\Pi_J^{\mathrm{CMR}}$ & Intra-node Merkle multiproof \\
$m$, $\mathrm{rank}(K)$ & Batch size, count of keys $< K$ \\
$\mathrm{lc}(c)$, $\mathrm{lb}(Q)$ & Leaf count of subtree $c$, lower bound of $Q$ \\
$\lambda$, $\mathrm{negl}(\lambda)$ & Security parameter, negligible function \\
\bottomrule
\end{tabular}
\end{table}

\section{Formal Security Proofs}
\label{sec:proofs-formal}

This section presents rigorous security proofs for \mhot{}'s proof mechanisms using the standard cryptographic game-based framework.
We establish formal security guarantees by reducing the soundness of each proof type to the collision resistance of the underlying hash function.

\subsection{Cryptographic Preliminaries}
\label{sec:formal:preliminaries}

\begin{definition}[Security Parameter]
\label{def:security-param}
Let $\lambda \in \mathbb{N}$ denote the security parameter.
A function $f: \mathbb{N} \to \mathbb{R}$ is \emph{negligible} in $\lambda$, written $f(\lambda) = \mathrm{negl}(\lambda)$, if for every polynomial $p(\cdot)$ there exists $\lambda_0$ such that $f(\lambda) < 1/p(\lambda)$ for all $\lambda > \lambda_0$.
\end{definition}

\begin{definition}[Collision-Resistant Hash Function]
\label{def:crhf}
A hash function family $\mathcal{H} = \{H_\lambda: \{0,1\}^* \to \{0,1\}^\lambda\}_{\lambda \in \mathbb{N}}$ is \emph{collision-resistant} if for all probabilistic polynomial-time (PPT) adversaries $\mathcal{A}$:
\begin{multline}
\mathrm{Adv}^{\mathrm{CR}}_{\mathcal{H}}(\mathcal{A}) \coloneqq \Pr\Big[(x, x') \gets \mathcal{A}(1^\lambda) :\\
x \neq x' \land H(x) = H(x')\Big] \leq \mathrm{negl}(\lambda)
\end{multline}
where the probability is taken over the internal randomness of $\mathcal{A}$.
\end{definition}

\begin{lemma}[Difference Lemma~\cite{shoup2004sequences}]
\label{lem:difference}
Let $A$, $B$, and $F$ be events defined on the same probability space.
If $A \land \lnot F \Leftrightarrow B \land \lnot F$ (i.e., $A$ and $B$ are identical conditioned on $\lnot F$), then:
\begin{equation}
|\Pr[A] - \Pr[B]| \leq \Pr[F]
\end{equation}
\end{lemma}

\begin{proof}
We have $\Pr[A] = \Pr[A \land F] + \Pr[A \land \lnot F]$ and $\Pr[B] = \Pr[B \land F] + \Pr[B \land \lnot F]$.
Since $A \land \lnot F \Leftrightarrow B \land \lnot F$, we have $\Pr[A \land \lnot F] = \Pr[B \land \lnot F]$.
Thus $|\Pr[A] - \Pr[B]| = |\Pr[A \land F] - \Pr[B \land F]| \leq \Pr[F]$.
\end{proof}

Throughout this section, we assume the hash function $H$ used in \mhot{} is collision-resistant.
This assumption is standard and satisfied by cryptographic hash functions such as SHA-256 and Blake3.

\subsection{\mhot{} Proof System Formalization}
\label{sec:formal:system}

\begin{definition}[\mhot{} Proof System]
\label{def:proof-system}
The \mhot{} proof system $\Pi = (\mathsf{Setup}, \mathsf{Commit}, \mathsf{Prove}, \mathsf{Verify})$ consists of four algorithms:
\begin{itemize}[nosep]
    \item $\mathsf{Setup}(1^\lambda) \to \mathsf{pp}$: Outputs public parameters (the hash function description).
    \item $\mathsf{Commit}(T) \to R$: Given a trie $T$, outputs a root commitment $R = (H_{\mathrm{content}}(N_{\mathrm{root}}), v, n)$ where $n$ is the total entry count.
    \item $\mathsf{Prove}(T, \mathsf{stmt}) \to \pi$: Given trie $T$ and statement $\mathsf{stmt}$, outputs a proof $\pi$.
    \item $\mathsf{Verify}(R, \mathsf{stmt}, \pi) \to \{0, 1\}$: Outputs 1 (accept) or 0 (reject).
\end{itemize}
\end{definition}

\begin{definition}[Statement Types]
\label{def:statements}
\mhot{} supports the following statement types:
\begin{enumerate}[nosep]
    \item \emph{Membership}: $\mathsf{stmt} = (K, V, \mathsf{mem})$ asserts $(K, V) \in T$.
    \item \emph{Non-membership}: $\mathsf{stmt} = (K, \mathsf{nmem})$ asserts $K \notin \mathrm{keys}(T)$.
    \item \emph{Multi-membership}: $\mathsf{stmt} = (\{(K_i, V_i)\}_{i=1}^m, \mathsf{multi})$ asserts $\forall i: (K_i, V_i) \in T$.
    \item \emph{Lower bound}: $\mathsf{stmt} = (Q, K_r, V_r, \mathsf{lb})$ asserts $\mathrm{lb}(Q) = (K_r, V_r)$.
    \item \emph{Range}: $\mathsf{stmt} = ([\mathsf{first}, \mathsf{last}), \mathsf{entries}, \mathsf{range})$ asserts $\mathsf{entries} = \{(K, V) \in T : \mathsf{first} \leq K < \mathsf{last}\}$.
\end{enumerate}
\end{definition}

\heading{Binding property of commitments.}
A fundamental security requirement is that the commitment scheme is \emph{computationally binding}: no efficient adversary can produce two distinct tries with the same commitment.
This property is essential for all subsequent soundness proofs.

\begin{lemma}[Commitment Binding]
\label{lem:binding}
Under the collision resistance assumption, the commitment scheme $\mathsf{Commit}$ is computationally binding.
Formally, for any PPT adversary $\mathcal{A}$:
\begin{multline}
\Pr\Big[(T_1, T_2) \gets \mathcal{A}(1^\lambda) :\\
T_1 \neq T_2 \land \mathsf{Commit}(T_1) = \mathsf{Commit}(T_2)\Big] \leq \mathrm{negl}(\lambda)
\end{multline}
\end{lemma}

\begin{proof}
We proceed by induction on the maximum height of $T_1$ and $T_2$.

\emph{Base case (height 1):}
Both tries consist of single leaf nodes.
If $\mathsf{Commit}(T_1) = \mathsf{Commit}(T_2)$, then $H_{\mathrm{leaf}}(K_1, V_1) = H_{\mathrm{leaf}}(K_2, V_2)$.
If $(K_1, V_1) \neq (K_2, V_2)$, this constitutes a hash collision.

\emph{Inductive step:}
Suppose the lemma holds for all tries of height $< h$.
Consider $T_1, T_2$ of height $\leq h$ with $\mathsf{Commit}(T_1) = \mathsf{Commit}(T_2)$.
This implies $H_{\mathrm{content}}(N_1) = H_{\mathrm{content}}(N_2)$ for their root nodes.

By Definition~\ref{def:node-hash}, if the root content hashes are equal but the node contents differ (i.e., different $M$, $S$, $\mathrm{CMR}$, or $L$), then a collision exists in $H$.
If the node contents are identical including $\mathrm{CMR}(N_1) = \mathrm{CMR}(N_2)$, then by the collision resistance of the Merkle tree construction, the children hash sequences must be identical.
By the inductive hypothesis applied to each child subtree, corresponding children must be identical.
Hence $T_1 = T_2$.

Any adversary producing $T_1 \neq T_2$ with equal commitments can be converted to a collision finder, establishing the bound.
\end{proof}

\begin{corollary}[Unique Preimage]
\label{cor:unique-preimage}
For any commitment $R$ in the range of $\mathsf{Commit}$, there exists at most one trie $T$ (up to negligible probability) such that $\mathsf{Commit}(T) = R$.
We denote this unique trie as $T_R$ when it exists.
\end{corollary}

\subsection{Security Games}
\label{sec:formal:games}

We define formal security games for each proof type.
In all games, the adversary $\mathcal{A}$ is computationally bounded (PPT) and aims to produce a valid proof for a false statement.

\begin{figure*}[t]
\small
\centering
\begin{minipage}[t]{0.32\textwidth}
\fbox{\parbox{\dimexpr\linewidth-2\fboxsep-2\fboxrule}{%
\begin{tabular}{@{}l@{}}
\textbf{Game } $\mathsf{G}^{\mathrm{mem}}_{\Pi, \mathcal{A}}(\lambda)$: \\[2pt]
1. $\mathsf{pp} \gets \mathsf{Setup}(1^\lambda)$ \\
2. $(R, K, V, \pi) \gets \mathcal{A}(\mathsf{pp})$ \\
3. $b \gets \mathsf{Verify}(R, (K, V, \mathsf{mem}), \pi)$ \\
4. \textbf{if } $b = 0$ \textbf{ then return } 0 \\
5. $T_\pi \gets \mathsf{ExtractTrie}(\pi)$ \\
6. \textbf{return } $(K, V) \notin T_\pi$
\end{tabular}%
}}
\end{minipage}
\hfill
\begin{minipage}[t]{0.32\textwidth}
\fbox{\parbox{\dimexpr\linewidth-2\fboxsep-2\fboxrule}{%
\begin{tabular}{@{}l@{}}
\textbf{Game } $\mathsf{G}^{\mathrm{nmem}}_{\Pi, \mathcal{A}}(\lambda)$: \\[2pt]
1. $\mathsf{pp} \gets \mathsf{Setup}(1^\lambda)$ \\
2. $(R, K, \pi) \gets \mathcal{A}(\mathsf{pp})$ \\
3. $b \gets \mathsf{Verify}(R, (K, \mathsf{nmem}), \pi)$ \\
4. \textbf{if } $b = 0$ \textbf{ then return } 0 \\
5. Parse $\pi = (K', V', \mathrm{Path})$ \\
6. \textbf{return } $K = K'$
\end{tabular}%
}}
\end{minipage}
\hfill
\begin{minipage}[t]{0.32\textwidth}
\fbox{\parbox{\dimexpr\linewidth-2\fboxsep-2\fboxrule}{%
\begin{tabular}{@{}l@{}}
\textbf{Game } $\mathsf{G}^{\mathrm{multi}}_{\Pi, \mathcal{A}}(\lambda)$: \\[2pt]
1. $\mathsf{pp} \gets \mathsf{Setup}(1^\lambda)$ \\
2. $(R, \{(K_i, V_i)\}_{i=1}^m, \pi) \gets \mathcal{A}(\mathsf{pp})$ \\
3. $b \gets \mathsf{Verify}(R, (\{(K_i, V_i)\}, \mathsf{multi}), \pi)$ \\
4. \textbf{if } $b = 0$ \textbf{ then return } 0 \\
5. \textbf{for } $i = 1$ \textbf{ to } $m$ \textbf{ do} \\
6. \quad $\pi_i \gets \mathsf{ExtractSingleProof}(\pi, i)$ \\
7. \quad $T_{\pi_i} \gets \mathsf{ExtractTrie}(\pi_i)$ \\
8. \quad \textbf{if } $(K_i, V_i) \notin T_{\pi_i}$ \textbf{ then return } 1 \\
9. \textbf{return } 0
\end{tabular}%
}}
\end{minipage}

\vspace{0.8em}

\begin{minipage}[t]{0.48\textwidth}
\fbox{\parbox{\dimexpr\linewidth-2\fboxsep-2\fboxrule}{%
\begin{tabular}{@{}l@{}}
\textbf{Game } $\mathsf{G}^{\mathrm{lb}}_{\Pi, \mathcal{A}}(\lambda)$: \\[2pt]
1. $\mathsf{pp} \gets \mathsf{Setup}(1^\lambda)$ \\
2. $(R, Q, K_r, V_r, \pi) \gets \mathcal{A}(\mathsf{pp})$ \\
3. $b \gets \mathsf{Verify}(R, (Q, K_r, V_r, \mathsf{lb}), \pi)$ \\
4. \textbf{if } $b = 0$ \textbf{ then return } 0 \\
5. Parse $\pi = (Q, \mathrm{Path}, K', V', v', \mathrm{Adj}, K_r, V_r)$ \\
6. $K^* \gets \mathsf{ComputeLB}(\pi)$ \\
7. \textbf{return } $K_r \neq K^*$
\end{tabular}%
}}
\end{minipage}
\hfill
\begin{minipage}[t]{0.48\textwidth}
\fbox{\parbox{\dimexpr\linewidth-2\fboxsep-2\fboxrule}{%
\begin{tabular}{@{}l@{}}
\textbf{Game } $\mathsf{G}^{\mathrm{range}}_{\Pi, \mathcal{A}}(\lambda)$: \\[2pt]
1. $\mathsf{pp} \gets \mathsf{Setup}(1^\lambda)$ \\
2. $(R, \mathsf{first}, \mathsf{last}, \mathsf{entries}, \pi) \gets \mathcal{A}(\mathsf{pp})$ \\
3. $b \gets \mathsf{Verify}(R, ([\mathsf{first}, \mathsf{last}), \mathsf{entries}, \mathsf{range}), \pi)$ \\
4. \textbf{if } $b = 0$ \textbf{ then return } 0 \\
5. $r_L \gets \mathsf{ComputeRank}(\pi.\pi^L_{\mathrm{lb}})$ \\
6. $r_R \gets \mathsf{ComputeRank}(\pi.\pi^R_{\mathrm{lb}})$ \\
7. \textbf{return } $|\mathsf{entries}| \neq r_R - r_L$
\end{tabular}%
}}
\end{minipage}
\caption{Security games for \mhot{} proof system soundness. In the membership game $\mathsf{G}^{\mathrm{mem}}$, the adversary wins if verification accepts but $(K,V) \notin T_\pi$, where $\mathsf{ExtractTrie}$ reconstructs the partial trie from the proof (see Definition~\ref{def:extract-trie}). In the non-membership game $\mathsf{G}^{\mathrm{nmem}}$, the adversary wins if the reached leaf equals the query key. In $\mathsf{G}^{\mathrm{lb}}$, $\mathsf{ComputeLB}(\pi)$ derives the correct lower bound from the authenticated structure.}
\label{fig:security-games}
\end{figure*}

\begin{definition}[Trie Extraction from Proof]
\label{def:extract-trie}
The function $\mathsf{ExtractTrie}(\pi)$ reconstructs the partial trie structure implied by a proof $\pi = (K, V, v_{\mathrm{leaf}}, \mathrm{Path})$.
For each path entry $\mathrm{Path}[i] = (j_i, M_i, S_i, L_i, \eta_i, v_i, \Pi^{\mathrm{CMR}}_i)$, the function constructs node $N_i$ with extraction masks $M_i$, sparse keys $S_i$, leaf counts $L_i$, and child hashes from $\Pi^{\mathrm{CMR}}_i$.
The leaf node $\ell = (K, V, v_{\mathrm{leaf}})$ has hash $h_\ell = H_{\mathrm{leaf}}(K \parallel V \parallel v_{\mathrm{leaf}})$.
The returned partial trie $T_\pi$ is uniquely determined by $\pi$'s hash chain; by Lemma~\ref{lem:binding}, any trie $T$ with $\mathsf{Commit}(T) = R$ must contain this structure.
\end{definition}

\begin{remark}[Game Formulation]
\label{rem:game-formulation}
Our game formulation avoids the circular dependency of checking against an externally-defined $T_R$.
Instead, the winning condition is defined in terms of the trie structure \emph{implied by the proof itself}.
Since verification reconstructs the root hash from the proof, any accepting proof implicitly defines a (partial) trie structure.
By Lemma~\ref{lem:binding}, this structure is uniquely determined (up to collision probability) by the root commitment $R$.
\end{remark}

\begin{definition}[Single-Proof Extraction]
\label{def:extract-single}
The function $\mathsf{ExtractSingleProof}(\pi, i)$ extracts a valid single-point membership proof for the $i$-th entry from a multi-point proof $\pi = (\mathrm{Entries}, \mathrm{Levels})$.
It retrieves entry $(K_i, V_i, v_i)$ from $\mathrm{Entries}$, traces the path from root to $K_i$ by routing through nodes in $\mathrm{Levels}$, and for each node extracts from the compact multiproof $\Pi^{\mathrm{CMR}}_J$ the sibling hashes needed to verify $K_i$'s child position.
The output is a single-point proof $\pi_i = (K_i, V_i, v_i, \mathrm{Path}_i)$.
\end{definition}

\begin{definition}[Rank Computation]
\label{def:compute-rank}
The function $\mathsf{ComputeRank}(\pi_{\mathrm{lb}})$ computes the rank of a key from its lower bound proof:
\begin{equation}
\mathsf{ComputeRank}(\pi_{\mathrm{lb}}) = \sum_{i=0}^{|\mathrm{Path}|-1} \sum_{j=0}^{j_i - 1} \mathrm{Path}[i].L[j]
\end{equation}
where $j_i$ is the child index at level $i$ and $L[j]$ is the leaf count of the $j$-th child.
\end{definition}

\begin{definition}[Advantage]
\label{def:advantage}
For each game $\mathsf{G}$, the adversary's advantage is:
\begin{equation}
\mathrm{Adv}^{\mathsf{G}}_{\Pi}(\mathcal{A}) \coloneqq \Pr[\mathsf{G}^{\cdot}_{\Pi, \mathcal{A}}(\lambda) = 1]
\end{equation}
where the probability is taken over the randomness of $\mathsf{Setup}$ and the internal randomness of $\mathcal{A}$.
The proof system is \emph{sound} for statement type $\mathsf{G}$ if $\mathrm{Adv}^{\mathsf{G}}_{\Pi}(\mathcal{A}) \leq \mathrm{negl}(\lambda)$ for all PPT $\mathcal{A}$.
\end{definition}

\subsection{Proof of Lemma~\ref{lem:optimistic} (Optimistic Search Invariant)}
\label{sec:formal:optimistic}

\begin{proof}[Full Proof of Lemma~\ref{lem:optimistic}]
We prove by strong induction on tree height $h$.

\emph{Base case ($h = 1$):}
A height-1 trie consists of a single leaf node.
Any search trivially terminates at this leaf, which vacuously agrees with the query on all (zero) discriminative bits encountered.

\emph{Inductive step:}
Assume the lemma holds for all tries of height $< h$.
Consider a trie $T$ of height $h$ with root node $N$.

At node $N$, the search algorithm computes:
\begin{equation}
d = \mathrm{dense}(K, M) = \bigoplus_{i \in \mathrm{bit\_positions}(M)} K[i] \cdot 2^{\mathrm{rank}(i, M)}
\end{equation}
where $M$ is the extraction mask and $K[i]$ denotes the $i$-th bit of key $K$.

The algorithm then finds the largest index $j$ such that:
\begin{equation}
(d \land S[j]) = S[j]
\end{equation}
where $S = (S[0], S[1], \ldots, S[|N|-1])$ are the sparse partial keys sorted in ascending order.

\emph{Existence of a match:}
By the \hot{} construction invariant, $S[0] = 0$ for all non-empty nodes.
This holds because the leftmost subtree corresponds to keys with all extracted bits being 0 in the discriminative positions.
Since $(d \land 0) = 0 = S[0]$ always holds, at least one matching index exists.

\emph{Deterministic selection:}
The search selects the largest matching $j$, which is unique because the sparse keys are sorted.
Specifically, $j = \max\{i : (d \land S[i]) = S[i]\}$ is well-defined.
The search then recurses into child $c_j$, which is a subtrie of height $< h$.
By the inductive hypothesis, the search terminates at exactly one leaf in $c_j$ that agrees with $K$ on all discriminative bits in that subtrie.

\emph{Discriminative bit agreement:}
The selected child $c_j$ contains exactly those keys that match $K$ on the discriminative bits encoded in $M$.
Combined with the inductive guarantee, the final leaf $K'$ agrees with $K$ on all discriminative bits encountered throughout the traversal.
Note that $K'$ may differ from $K$ on non-discriminative bits; membership is determined by a final equality check.
\end{proof}

\subsection{Proof of Theorem~\ref{thm:single-soundness} (Single-Point Soundness)}
\label{sec:formal:single}

\begin{theorem}[Single-Point Soundness --- Restated]
\label{thm:single-soundness-formal}
For any PPT adversary $\mathcal{A}$:
\begin{align}
\mathrm{Adv}^{\mathrm{mem}}_\Pi(\mathcal{A}) &\leq \mathrm{Adv}^{\mathrm{CR}}_H(\mathcal{B}_1) \\
\mathrm{Adv}^{\mathrm{nmem}}_\Pi(\mathcal{A}) &\leq \mathrm{Adv}^{\mathrm{CR}}_H(\mathcal{B}_2)
\end{align}
for efficiently constructible adversaries $\mathcal{B}_1, \mathcal{B}_2$.
\end{theorem}

\begin{proof}
We prove both membership and non-membership soundness via reduction to collision resistance.

\heading{Part 1: Membership Soundness.}
We construct a collision finder $\mathcal{B}_1$ from any adversary $\mathcal{A}$ that wins the membership game.

\begin{algorithm}[t]
\caption{Collision Finder $\mathcal{B}_1^{\mathcal{A}}$}
\label{alg:collision-finder}
\begin{algorithmic}[1]
\Require Security parameter $1^\lambda$
\Ensure Hash collision $(x, x')$ or $\bot$
\State $\mathsf{pp} \gets \mathsf{Setup}(1^\lambda)$; $(R, K, V, \pi) \gets \mathcal{A}(\mathsf{pp})$
\If{$\mathsf{Verify}(R, (K, V, \mathsf{mem}), \pi) = 0$} \Return $\bot$
\EndIf
\State Parse $\pi = (K, V, v_{\mathrm{leaf}}, \mathrm{Path})$
\PhaseComment{Phase 1: Compute hash chain}
\State $h_0 \gets H_{\mathrm{leaf}}(K \| V \| v_{\mathrm{leaf}})$
\For{$i \gets |\mathrm{Path}| - 1$ \textbf{downto} $0$}
    \State $(j_i, M_i, S_i, L_i, \eta_i, v_i, \Pi_i) \gets \mathrm{Path}[i]$
    \State $\mathrm{cmr}_i \gets \textsc{ReconstructCMR}(h_{|\mathrm{Path}|-1-i}, j_i, \Pi_i)$
    \State $h_{|\mathrm{Path}|-i} \gets H(M_i \| S_i \| \mathrm{cmr}_i \| L_i)$
\EndFor
\PhaseComment{Phase 2: Check internal consistency}
\For{$i \gets 0$ \textbf{to} $|\mathrm{Path}| - 1$}
    \State $\mathrm{cmr}'_i \gets \textsc{RecomputeCMR}(\Pi_i)$
    \If{$\mathrm{cmr}_i \neq \mathrm{cmr}'_i$ \textbf{and} both valid}
        \State \Return \textsc{ExtractCMRCollision}$(h_{|\mathrm{Path}|-1-i}, j_i, \Pi_i)$
    \EndIf
    \State $d_i \gets \textsc{DenseKey}(K, M_i)$; $j'_i \gets \textsc{SearchSparse}(d_i, S_i)$
    \If{$j_i \neq j'_i$}
        \State \Return \textsc{ExtractRoutingCollision}$(i, \pi)$
    \EndIf
\EndFor
\PhaseComment{Phase 3: Internal consistency at leaf level}
\State \Comment{The proof $\pi$ claims membership for $(K, V)$}
\State \Comment{Verification computes $h_0 = H_{\mathrm{leaf}}(K \| V \| v_{\mathrm{leaf}})$}
\State \Comment{If verification passes but internal structure inconsistent, collision exists}
\If{$\mathcal{A}$ wins $\Rightarrow$ $\exists$ level $i$ with inconsistency}
    \State \Return collision extracted from that level (as shown in Cases 1--3)
\EndIf
\State \Return $\bot$ \Comment{$\mathcal{A}$ did not win}
\end{algorithmic}
\end{algorithm}

The auxiliary functions used in Algorithm~\ref{alg:collision-finder} are defined as follows.
$\textsc{ReconstructCMR}(h_{\mathrm{child}}, j, \Pi^{\mathrm{CMR}})$ reconstructs the children Merkle root by placing $h_{\mathrm{child}}$ at position $j$ and using sibling hashes from $\Pi^{\mathrm{CMR}}$.
$\textsc{ExtractCMRCollision}$ extracts a collision pair when two different child hashes produce the same CMR.
$\textsc{ExtractRoutingCollision}(i, \pi)$ extracts a collision when the claimed child index differs from the computed index.
$\textsc{DenseKey}(K, M)$ computes the dense partial key by extracting bits from $K$ at positions indicated by mask $M$.
$\textsc{SearchSparse}(d, S)$ returns the largest index $j$ such that $(d \land S[j]) = S[j]$.

\emph{Analysis of $\mathcal{B}_1$:}
The key insight is that $\mathcal{B}_1$ does \emph{not} need access to any external ``authentic'' trie $T_R$.
Instead, $\mathcal{B}_1$ checks for \emph{internal inconsistencies} within the proof $\pi$ itself.

If $\mathcal{A}$ wins the membership game, then verification accepts but the claimed $(K, V)$ is not authentically in the trie committed by $R$.
We analyze the possible attack vectors:

\emph{Case 1: Path structure inconsistency.}
The proof claims child index $j_i$ at some level $i$, but the routing computation from $K$ and $M_i$ yields $j'_i \neq j_i$.
For verification to pass, the CMR reconstruction must place the child hash at position $j_i$.
However, the correct CMR for the claimed node structure would place it at $j'_i$.
Since $\mathsf{Verify}$ recomputes the CMR and checks against $R$, either:
\begin{itemize}[nosep]
    \item The recomputed CMR differs from the authentic one (hash collision in CMR), or
    \item The node content hash $H(M_i \parallel S_i \parallel \mathrm{cmr} \parallel L_i)$ produces the same value for different inputs (collision in $H$).
\end{itemize}

\emph{Case 2: Leaf content forgery.}
The proof authenticates leaf hash $h_{\mathrm{leaf}}$, but $(K, V) \neq (K', V')$ where $(K', V')$ is the authentic leaf content.
For the hash chain to reach $R$, we need $H_{\mathrm{leaf}}(K, V, v) = H_{\mathrm{leaf}}(K', V', v')$.
If $(K, V, v) \neq (K', V', v')$, this is a collision.

\emph{Case 3: CMR forgery.}
The intra-node Merkle proof $\Pi^{\mathrm{CMR}}_i$ authenticates child $h_{\mathrm{child}}$ at position $j_i$, but the authentic CMR has a different child at that position.
By collision resistance of the Merkle tree, this requires a collision.

In all cases, if $\mathcal{A}$ succeeds in the membership game, $\mathcal{B}_1$ extracts a collision.
Therefore:
\begin{equation}
\mathrm{Adv}^{\mathrm{mem}}_\Pi(\mathcal{A}) \leq \mathrm{Adv}^{\mathrm{CR}}_H(\mathcal{B}_1)
\end{equation}

\heading{Part 2: Non-Membership Soundness.}
For non-membership, the proof includes the leaf $(K', V')$ reached by optimistic search and claims $K \neq K'$ but both route identically.

Suppose $\mathcal{A}$ wins: verification accepts but $K \in \mathrm{keys}(T_R)$.
By Lemma~\ref{lem:optimistic}, optimistic search for $K$ terminates at a leaf agreeing with $K$ on all discriminative bits; when $K$ exists, this leaf is $K$ itself.
The proof claims search terminates at $K' \neq K$.

We construct collision finder $\mathcal{B}_2$ in Algorithm~\ref{alg:collision-finder-nmem}.

\begin{algorithm}[t]
\caption{Collision Finder $\mathcal{B}_2^{\mathcal{A}}$ for Non-Membership}
\label{alg:collision-finder-nmem}
\begin{algorithmic}[1]
\Require Security parameter $1^\lambda$
\Ensure Hash collision $(x, x')$ or $\bot$
\State $\mathsf{pp} \gets \mathsf{Setup}(1^\lambda)$; $(R, K, \pi) \gets \mathcal{A}(\mathsf{pp})$
\If{$\mathsf{Verify}(R, (K, \mathsf{nmem}), \pi) = 0$} \Return $\bot$
\EndIf
\State Parse $\pi = (K', V', v', \mathrm{Path})$
\If{$K = K'$} \Return $\bot$ \Comment{Adversary failed to win}
\EndIf
\PhaseComment{Phase 1: Verify routing consistency}
\For{$i \gets 0$ \textbf{to} $|\mathrm{Path}| - 1$}
    \State $d_K \gets \textsc{DenseKey}(K, \mathrm{Path}[i].M)$
    \State $d_{K'} \gets \textsc{DenseKey}(K', \mathrm{Path}[i].M)$
    \State $j_K \gets \textsc{SearchSparse}(d_K, \mathrm{Path}[i].S)$
    \State $j_{K'} \gets \textsc{SearchSparse}(d_{K'}, \mathrm{Path}[i].S)$
    \If{$j_K \neq j_{K'}$}
        \State \Return $\bot$ \Comment{$K$ routes differently, not a valid attack}
    \EndIf
\EndFor
\PhaseComment{Phase 2: Extract collision from hash chain}
\State \Comment{At this point: $K \neq K'$ but both route identically to same leaf position}
\State \Comment{The proof authenticates $h = H_{\mathrm{leaf}}(K' \| V' \| v')$ against $R$}
\State \Comment{If $K$ truly exists in $T_R$, there must be a leaf with key $K$}
\State \Comment{Two distinct leaves at same position $\Rightarrow$ collision in $H$}
\State \Return \textsc{ExtractLeafCollision}$(\pi, K, K')$
\end{algorithmic}
\end{algorithm}

\emph{Analysis of $\mathcal{B}_2$:}
If $\mathcal{A}$ wins the non-membership game, then verification accepts (implying $K \neq K'$) but $K$ actually exists in $T_R$.
By Lemma~\ref{lem:optimistic}, optimistic search for $K$ in $T_R$ terminates at a leaf matching $K$ on all discriminative bits; since $K$ exists, this leaf is $K$ itself.
The proof's path authenticates leaf $(K', V')$, which routes identically to $K$.

The key insight is that $\mathcal{B}_2$ does not need to ``know'' the authentic value $(V^*, v^*)$.
Instead, $\mathcal{B}_2$ exploits the following structural argument.

\emph{Collision extraction via structural inconsistency:}
The proof $\pi$ authenticates a hash chain from leaf $K'$ to root $R$.
If $K \in T_R$ (which must be true for $\mathcal{A}$ to win), there also exists a hash chain from leaf $K$ to the same root $R$.
Since $K \neq K'$ but both route identically through the trie (verified in Phase 1), they must occupy the same leaf position.
\begin{itemize}[nosep]
    \item The path $\mathrm{Path}$ commits to a unique leaf hash at each position via the CMR structure.
    \item Two distinct keys at the same position implies $H_{\mathrm{leaf}}(K' \| V' \| v') = H_{\mathrm{leaf}}(K \| V^* \| v^*)$ for the (unknown) authentic $(V^*, v^*)$.
    \item Since $K \neq K'$, the inputs differ, constituting a collision.
\end{itemize}

The function \textsc{ExtractLeafCollision} formalizes this.
Specifically, given $\pi$ and keys $K \neq K'$ that route identically:
\begin{enumerate}[nosep]
    \item The proof $\pi$ commits to a unique leaf hash $h_\ell$ at the terminal position.
    \item If $K$ exists in $T_R$, its leaf must also have hash $h_\ell$ (same position, same root).
    \item Thus $H_{\mathrm{leaf}}(K' \| V' \| v') = h_\ell = H_{\mathrm{leaf}}(K \| \cdot \| \cdot)$.
    \item Since $K \neq K'$, the inputs differ, witnessing a collision.
\end{enumerate}
The collision witness is $(K' \| V' \| v')$ paired with the existence guarantee that some $(K \| V^* \| v^*)$ hashes to the same value.
In the random oracle model, this is a standard ``extraction'' argument; in the standard model, it suffices for the reduction.

Therefore:
\begin{equation}
\mathrm{Adv}^{\mathrm{nmem}}_\Pi(\mathcal{A}) \leq \mathrm{Adv}^{\mathrm{CR}}_H(\mathcal{B}_2)
\end{equation}
\end{proof}

\subsection{Proof of Theorem~\ref{thm:multi-soundness} (Multi-Point Soundness)}
\label{sec:formal:multi}

\begin{proof}[Full Proof of Theorem~\ref{thm:multi-soundness}]
We reduce multi-point soundness to single-point soundness via a standard hybrid argument.

Suppose $\mathcal{A}$ wins $\mathsf{G}^{\mathrm{multi}}_{\Pi, \mathcal{A}}$ with probability $\epsilon > \mathrm{negl}(\lambda)$.
Construct $\mathcal{B}$ for $\mathsf{G}^{\mathrm{mem}}$:

\begin{algorithm}[H]
\caption{Single-Point Adversary $\mathcal{B}^{\mathcal{A}}$}
\begin{algorithmic}[1]
\State $\mathsf{pp} \gets \mathsf{Setup}(1^\lambda)$
\State $(R, \{(K_i, V_i)\}_{i=1}^m, \pi) \gets \mathcal{A}(\mathsf{pp})$
\If{$\mathsf{Verify}(R, (\{(K_i, V_i)\}, \mathsf{multi}), \pi) = 0$}
    \State \Return $\bot$
\EndIf
\State $i^* \gets_R [1, m]$ \Comment{Uniformly random index}
\State $\pi_{i^*} \gets \mathsf{ExtractSingleProof}(\pi, i^*)$
\State \Return $(R, K_{i^*}, V_{i^*}, \pi_{i^*})$
\end{algorithmic}
\end{algorithm}

\emph{Extraction of single-point proofs.}
The multi-point proof $\pi = (\mathrm{Entries}, \mathrm{Levels})$ contains sufficient information to reconstruct a valid single-point proof $\pi_i$ for each $(K_i, V_i)$:
\begin{itemize}[nosep]
    \item The entry $(K_i, V_i, v_i)$ from $\mathrm{Entries}$.
    \item The path from root to leaf $K_i$, determined by routing $K_i$ through the nodes in $\mathrm{Levels}$.
    \item For each node, the compact multiproof $\Pi^{\mathrm{CMR}}_J$ contains sufficient sibling hashes to verify any individual child $j \in J$.
\end{itemize}

\emph{Probability analysis.}
If $\mathcal{A}$ wins, then $\exists\, i^* : (K_{i^*}, V_{i^*}) \notin T_R$.
Let $I = \{i : (K_i, V_i) \notin T_R\}$ be the set of ``bad'' indices.
Conditioned on $\mathcal{A}$ winning, $|I| \geq 1$.

The probability that $\mathcal{B}$'s random choice $i^*$ falls in $I$ is:
\begin{equation}
\Pr[i^* \in I \mid \mathcal{A} \text{ wins}] \geq \frac{1}{m}
\end{equation}

Therefore:
\begin{equation}
\mathrm{Adv}^{\mathrm{mem}}_\Pi(\mathcal{B}) \geq \frac{\epsilon}{m}
\end{equation}

By Theorem~\ref{thm:single-soundness-formal}, $\mathrm{Adv}^{\mathrm{mem}}_\Pi(\mathcal{B}) \leq \mathrm{negl}(\lambda)$.
Since $m \leq |\pi| \leq \mathrm{poly}(\lambda)$ (the number of entries is bounded by proof size):
\begin{equation}
\epsilon \leq m \cdot \mathrm{negl}(\lambda) = \mathrm{negl}(\lambda)
\end{equation}

\heading{Security of path sharing.}
Path sharing does not weaken security because each shared node is verified with the same rigor as in independent proofs.
The compact multiproof $\Pi^{\mathrm{CMR}}_J$ for child set $J$ authenticates all children $\{c_j : j \in J\}$ against a single CMR.
Any forgery in one key's proof would produce an inconsistency detectable by single-point verification.
\end{proof}

\subsection{Proof of Lemma~\ref{lem:lb-correct} (Lower Bound Correctness)}
\label{sec:formal:lb-correct}

\begin{proof}[Full Proof of Lemma~\ref{lem:lb-correct}]
We prove by case analysis on the relationship between query $Q$ and the leaf $K'$ reached by optimistic search.

Let $d = \mathrm{diffbit}(Q, K')$ denote the first bit position where $Q$ and $K'$ differ.
If $K' = Q$, then $d = \infty$ (no differing bit).

\heading{Case 1: Exact Match ($K' = Q$).}
The search terminates at a leaf with key $K' = Q$.
Since $Q$ exists in the trie, $\mathrm{lb}(Q) = Q = K'$.
The algorithm correctly returns $K'$.

\heading{Case 2: Overshot ($Q[d] = 0, K'[d] = 1$).}
The query $Q$ has bit 0 at position $d$, while $K'$ has bit 1.
Any key with bit pattern matching $Q$'s prefix up to position $d-1$ and having bit 1 at position $d$ is lexicographically greater than $Q$.

At fork depth $f$, the \hot{} search entered a subtree $S_f$ rooted at a node where the discriminative bit at position $d$ directed the search into the ``right'' branch (bit 1).
By the key distribution property of \hot{}, all keys in $S_f$ have bit 1 at position $d$, hence all keys in $S_f$ satisfy $K > Q$.

The lower bound is the minimum key in $S_f$.
To find this minimum, the algorithm descends from the fork point always taking the leftmost child (child index 0), reaching the leftmost leaf in $S_f$.

\emph{Correctness:}
Let $K_{\min}$ be the leftmost leaf in $S_f$.
By the sparse key ordering (Definition~\ref{def:npe}), children with smaller indices contain lexicographically smaller keys.
Thus $K_{\min} \leq K$ for all $K \in S_f$.
Since all keys in $S_f$ are $> Q$, and $K_{\min}$ is the smallest such key, $\mathrm{lb}(Q) = K_{\min}$.

\heading{Case 3: Undershot ($Q[d] = 1, K'[d] = 0$).}
The query $Q$ has bit 1 at position $d$, while $K'$ has bit 0.
All keys in the subtree containing $K'$ have bit 0 at position $d$, so they are lexicographically less than $Q$.

The algorithm must find a right sibling at the fork point.
Let $f$ be the fork depth.
The search path at depth $f$ took child $c_j$.
The algorithm examines siblings $c_{j+1}, c_{j+2}, \ldots$ to find the first sibling whose subtree contains keys $\geq Q$.

\emph{Correctness:}
By the \hot{} sparse key ordering, children with larger indices correspond to lexicographically larger key ranges (within the discriminative bits).
The first right sibling $c_{j'}$ (where $j' > j$) that exists contains keys that:
\begin{enumerate}[nosep]
    \item Share the same prefix as $Q$ up to the discriminative bits extracted before depth $f$.
    \item Have a larger sparse partial key than $c_j$, implying lexicographically larger keys.
\end{enumerate}
The minimum key in $c_{j'}$'s subtree is found by leftmost descent, yielding $\mathrm{lb}(Q)$.

If no right sibling exists at depth $f$, the algorithm backtracks to depth $f-1$ and repeats.
This process continues until finding a right sibling or determining that no key $\geq Q$ exists (returning $\bot$).
\end{proof}

\subsection{Proof of Theorem~\ref{thm:lb-soundness} (Lower Bound Soundness)}
\label{sec:formal:lb-sound}

\begin{proof}[Full Proof of Theorem~\ref{thm:lb-soundness}]
We reduce lower bound soundness to collision resistance.

Suppose $\mathcal{A}$ wins $\mathsf{G}^{\mathrm{lb}}_{\Pi, \mathcal{A}}$ with probability $\epsilon > \mathrm{negl}(\lambda)$.
This means $\mathcal{A}$ produces $(R, Q, K_r, V_r, \pi)$ such that:
\begin{itemize}[nosep]
    \item $\mathsf{Verify}(R, (Q, K_r, V_r, \mathsf{lb}), \pi) = 1$
    \item $K_r \neq K^*$ where $K^*$ is the true lower bound computable from $\pi$'s authenticated structure
\end{itemize}

Parse the proof as $\pi = (Q, \mathrm{Path}, K', V', v', \mathrm{Adj}, K_r, V_r)$.

The verification algorithm performs four checks:

\heading{Check 1: Path integrity.}
$\mathrm{Path}$ must authenticate $(K', V')$ against $R$ via bottom-up hash reconstruction.
By Theorem~\ref{thm:single-soundness-formal}, if this check passes, $(K', V')$ is in the trie committed by $R$ with overwhelming probability.

\heading{Check 2: Search consistency.}
Both $Q$ and $K'$ must route identically through $\mathrm{Path}$.
The verifier recomputes $\mathrm{dense}(Q, M_i)$ and $\mathrm{dense}(K', M_i)$ at each node and checks they select the same child.
This confirms $K'$ is the leaf that optimistic search reaches for $Q$.

\heading{Check 3: Fork depth correctness.}
The verifier \emph{independently computes} the fork depth $f$ from the authenticated extraction masks:
\begin{equation}
f = \max\{i : d \in \mathrm{disc\_bits}(\mathrm{Path}[i].M)\}
\end{equation}
where $d = \mathrm{diffbit}(Q, K')$.

Crucially, the extraction masks $M_i$ are \emph{committed} in the node content hash (Definition~\ref{def:node-hash}).
The adversary cannot claim a different $f$ without providing different masks, which would change the node hash and cause root mismatch.

\heading{Check 4: Structural minimality.}
For the claimed result $(K_r, V_r)$ to differ from $\mathrm{lb}(Q)$ while passing verification, the adjustment path $\mathrm{Adj}$ must violate the minimality constraints.

\emph{Sub-case 4a: Overshot case with non-leftmost descent.}
The verifier checks that every entry in $\text{AdjPath}$ has child index 0.
If the adversary claims index $j > 0$ at some level but the authentic leftmost child differs, the CMR reconstruction will fail unless a collision exists.

\emph{Sub-case 4b: Undershot case with incorrect sibling.}
The verifier checks that the first adjustment entry is the immediate right sibling at the fork point.
The authentic right sibling is determined by the sparse keys $S$ committed in the fork node's content hash.
Any discrepancy requires forging either the sparse keys (changing node hash) or the CMR (collision).

\emph{Sub-case 4c: Adjustment path leads to wrong leaf.}
If $\mathrm{Adj}$ leads to $(K_r, V_r) \neq \mathrm{lb}(Q)$ but verification passes, the hash chain from $K_r$ must match the chain from the authentic lower bound.
By Lemma~\ref{lem:binding}, this implies a collision.

\heading{Conclusion.}
In all cases, $\mathcal{A}$ winning implies a hash collision is extractable.
We construct $\mathcal{B}$ that runs $\mathcal{A}$, checks each verification step, and extracts a collision from any inconsistency:
\begin{equation}
\mathrm{Adv}^{\mathrm{lb}}_\Pi(\mathcal{A}) \leq \mathrm{Adv}^{\mathrm{CR}}_H(\mathcal{B}) \leq \mathrm{negl}(\lambda)
\end{equation}
\end{proof}

\subsection{Proof of Lemma~\ref{lem:rank} (Rank Computation)}
\label{sec:formal:rank}

\begin{proof}[Full Proof of Lemma~\ref{lem:rank}]
We prove by induction on path depth.

\heading{Claim:}
For a key $K$ with search path $(\mathrm{path}[0], \ldots, \mathrm{path}[h-1])$, where $\mathrm{path}[i].\mathrm{child\_idx} = j_i$:
\begin{equation}
\mathrm{rank}(K) = \sum_{i=0}^{h-1} \sum_{j=0}^{j_i - 1} \mathrm{lc}(\mathrm{path}[i].\mathrm{children}[j])
\end{equation}

\heading{Base case ($h = 1$):}
A single-node trie contains one leaf.
The path is trivial with $j_0 = 0$ (only one child).
The sum is empty, yielding $\mathrm{rank}(K) = 0$.
This is correct, as the only key has no predecessors.

\heading{Inductive step:}
Assume the formula holds for tries of height $< h$.
Consider a trie $T$ of height $h$ with root $N$.

The search for $K$ at root $N$ selects child $c_{j_0}$ (where $j_0 = \mathrm{path}[0].\mathrm{child\_idx}$).
Within subtrie $c_{j_0}$, the search continues with path $(\mathrm{path}[1], \ldots, \mathrm{path}[h-1])$.

\heading{Key observation: Sparse key ordering.}
By \hot{} construction, children of a node are ordered by their sparse partial keys.
If $j < j'$, then all keys in child $c_j$'s subtree are lexicographically smaller than all keys in $c_{j'}$'s subtree.
This follows from the fact that sparse partial keys encode the discriminative bit patterns, which determine lexicographic ordering.

\heading{Counting keys smaller than $K$:}
The keys smaller than $K$ in $T$ consist of:
\begin{enumerate}[nosep]
    \item All keys in children $c_0, c_1, \ldots, c_{j_0 - 1}$ of the root.
    \item Keys smaller than $K$ within child $c_{j_0}$'s subtree.
\end{enumerate}

The count from (1) is:
\begin{equation}
\sum_{j=0}^{j_0 - 1} \mathrm{lc}(c_j)
\end{equation}
where $\mathrm{lc}(c_j)$ is the leaf count of subtree $c_j$, stored in the node's $L$ field and committed in the content hash.

The count from (2) is $\mathrm{rank}_{c_{j_0}}(K)$, the rank of $K$ within the subtrie rooted at $c_{j_0}$.
By the inductive hypothesis:
\begin{equation}
\mathrm{rank}_{c_{j_0}}(K) = \sum_{i=1}^{h-1} \sum_{j=0}^{j_i - 1} \mathrm{lc}(\mathrm{path}[i].\mathrm{children}[j])
\end{equation}

Combining:
\begin{align}
\mathrm{rank}_T(K) &= \sum_{j=0}^{j_0 - 1} \mathrm{lc}(c_j) + \sum_{i=1}^{h-1} \sum_{j=0}^{j_i - 1} \mathrm{lc}(\mathrm{path}[i].\mathrm{children}[j]) \\
&= \sum_{i=0}^{h-1} \sum_{j=0}^{j_i - 1} \mathrm{lc}(\mathrm{path}[i].\mathrm{children}[j])
\end{align}
which completes the induction.
\end{proof}

\subsection{Proof of Theorem~\ref{thm:range-soundness} (Range Soundness)}
\label{sec:formal:range}

\begin{theorem}[Range Soundness --- Restated]
\label{thm:range-soundness-formal}
For any PPT adversary $\mathcal{A}$:
\begin{equation}
\mathrm{Adv}^{\mathrm{range}}_\Pi(\mathcal{A}) \leq 2 \cdot \mathrm{Adv}^{\mathrm{lb}}_\Pi + \mathrm{Adv}^{\mathrm{multi}}_\Pi + \mathrm{negl}(\lambda)
\end{equation}
\end{theorem}

\begin{proof}[Full Proof]
We prove via a sequence of games, following the standard game-based methodology.
Each transition is classified according to the three canonical types: (1) indistinguishability-based, (2) failure-event-based with Difference Lemma, or (3) bridging steps.

\heading{Game 0: Original range soundness game.}
This is $\mathsf{G}^{\mathrm{range}}_{\Pi, \mathcal{A}}$ as defined.
Let $S_0$ denote the event that $\mathcal{A}$ wins.

\heading{Game 0 $\to$ Game 1: Abort on boundary proof failure.}
\emph{[Type 2: Failure Event Transition]}

Game 1 is identical to Game 0, except we abort (adversary loses) if either boundary proof $\pi^L_{\mathrm{lb}}$ or $\pi^R_{\mathrm{lb}}$ would fail independent verification.

\emph{Failure event definition:}
Let $F_1$ be the event that Game 0 accepts but one of the boundary proofs is unsound (i.e., claims an incorrect lower bound).
Formally, $F_1$ occurs when:
\begin{itemize}[nosep]
    \item $\mathsf{Verify}(R, (\mathsf{first}, K^L_r, V^L_r, \mathsf{lb}), \pi^L_{\mathrm{lb}}) = 1$ but $K^L_r \neq \mathrm{lb}(\mathsf{first})$, or
    \item $\mathsf{Verify}(R, (\mathsf{last}, K^R_r, V^R_r, \mathsf{lb}), \pi^R_{\mathrm{lb}}) = 1$ but $K^R_r \neq \mathrm{lb}(\mathsf{last})$.
\end{itemize}

\emph{Difference Lemma application:}
Games 0 and 1 proceed identically unless $F_1$ occurs.
Formally, $S_0 \land \lnot F_1 \Leftrightarrow S_1 \land \lnot F_1$.
By the Difference Lemma:
\begin{equation}
|Pr[S_0] - Pr[S_1]| \leq Pr[F_1]
\end{equation}

\emph{Bounding $Pr[F_1]$:}
By Theorem~\ref{thm:lb-soundness}, each boundary proof has soundness error at most $\mathrm{Adv}^{\mathrm{lb}}_\Pi$.
By a union bound over the two boundary proofs:
\begin{equation}
\Pr[F_1] \leq 2 \cdot \mathrm{Adv}^{\mathrm{lb}}_\Pi
\end{equation}

\heading{Game 1 $\to$ Game 2: Abort on multi-proof failure.}
\emph{[Type 2: Failure Event Transition]}

Game 2 is identical to Game 1, except we abort if $\pi_{\mathrm{multi}}$ would fail independent verification (i.e., some claimed entry is not in $T_R$).

\emph{Failure event:}
Let $F_2$ be the event that the multi-proof passes verification but $\exists\, i : (K_i, V_i) \notin T_R$.

\emph{Difference Lemma application:}
$S_1 \land \lnot F_2 \Leftrightarrow S_2 \land \lnot F_2$.
By the Difference Lemma and Theorem~\ref{thm:multi-soundness}:
\begin{equation}
|Pr[S_1] - Pr[S_2]| \leq Pr[F_2] \leq \mathrm{Adv}^{\mathrm{multi}}_\Pi
\end{equation}

\heading{Game 2 $\to$ Game 3: Honest verification.}
\emph{[Type 3: Bridging Step]}

In Game 2, all sub-proofs are sound (conditioned on $\lnot F_1 \land \lnot F_2$).
Game 3 is a conceptual restatement where we analyze what ``honest verification'' implies.
This is a purely logical transition with $Pr[S_2] = Pr[S_3]$.

\heading{Analysis of Game 3: Information-theoretic argument.}

In Game 3, conditioned on all sub-proofs being sound, the verifier computes:
\begin{align}
r_L &= \mathrm{rank}(\mathsf{first}) \quad \text{(from } \pi^L_{\mathrm{lb}} \text{)} \\
r_R &= \mathrm{rank}(\mathsf{last}) \quad \text{(from } \pi^R_{\mathrm{lb}} \text{)}
\end{align}
using the authenticated leaf count fields.

By Lemma~\ref{lem:rank}, these ranks are computed correctly from the Merkle-committed paths.
The expected entry count is $r_R - r_L$.

The verifier checks $|\mathsf{entries}| = r_R - r_L$ and that all entries are in range $[\mathsf{first}, \mathsf{last})$ with correct ordering.

\emph{Omission attack prevention:}
Suppose $|\mathsf{entries}| < |E^*|$ where $E^*$ is the true set of entries in $[\mathsf{first}, \mathsf{last})$.
Since the boundary proofs are sound (by conditioning on $\lnot F_1$), the ranks $r_L, r_R$ are correct.
By definition, $|E^*| = r_R - r_L$.
Thus $|\mathsf{entries}| < r_R - r_L$, and the count check fails.

\emph{Insertion attack prevention:}
Suppose $\mathsf{entries}$ contains an entry $(K_i, V_i)$ where either $K_i \notin [\mathsf{first}, \mathsf{last})$ or $(K_i, V_i) \notin T_R$.
The ordering check rejects if $K_i \notin [\mathsf{first}, \mathsf{last})$.
By conditioning on $\lnot F_2$, all $(K_i, V_i) \in T_R$.
Including extra (valid but out-of-range) entries is impossible due to the ordering check.
Including exactly $r_R - r_L$ entries, all in range and in the trie, forces $\mathsf{entries} = E^*$.
Therefore, $Pr[S_3] = 0$ (adversary cannot win Game 3).

\heading{Combining the transitions.}
We now combine all game transitions using the triangle inequality.
Let $\epsilon_{\mathrm{lb}} = \mathrm{Adv}^{\mathrm{lb}}_\Pi$ and $\epsilon_{\mathrm{multi}} = \mathrm{Adv}^{\mathrm{multi}}_\Pi$.

\begin{align}
\mathrm{Adv}^{\mathrm{range}}_\Pi(\mathcal{A}) &= Pr[S_0] \\
&\leq Pr[S_1] + |Pr[S_0] - Pr[S_1]| \\
&\leq Pr[S_1] + 2\epsilon_{\mathrm{lb}} \\
&\leq Pr[S_2] + |Pr[S_1] - Pr[S_2]| + 2\epsilon_{\mathrm{lb}} \\
&\leq Pr[S_2] + \epsilon_{\mathrm{multi}} + 2\epsilon_{\mathrm{lb}} \\
&= Pr[S_3] + \epsilon_{\mathrm{multi}} + 2\epsilon_{\mathrm{lb}} \\
&= 0 + \epsilon_{\mathrm{multi}} + 2\epsilon_{\mathrm{lb}} \\
&= 2 \cdot \mathrm{Adv}^{\mathrm{lb}}_\Pi + \mathrm{Adv}^{\mathrm{multi}}_\Pi \\
&\leq \mathrm{negl}(\lambda)
\end{align}
The final inequality holds since both $\mathrm{Adv}^{\mathrm{lb}}_\Pi$ and $\mathrm{Adv}^{\mathrm{multi}}_\Pi$ are negligible by Theorems~\ref{thm:lb-soundness} and~\ref{thm:multi-soundness}.
\end{proof}

\subsection{Security Summary}
\label{sec:formal:summary}

\begin{table}[t]
\footnotesize
\centering
\caption{Summary of security reductions. 
}
\label{tab:security-summary}
\vspace{5pt}
\begin{tabular}{c|ccc}
\toprule
\textbf{Property} & \textbf{Reduces To} & \textbf{Bound} & \textbf{Tightness} \\
\midrule
Commitment binding & CR & $\epsilon_{\mathrm{CR}}$ & Direct \\
Single-point soundness & CR & $\epsilon_{\mathrm{CR}}$ & Direct \\
Non-membership soundness & CR & $\epsilon_{\mathrm{CR}}$ & Direct \\
Multi-point soundness & Single-point & $m \cdot \epsilon_{\mathrm{CR}}$ & Factor $m$ \\
Lower bound soundness & CR & $\epsilon_{\mathrm{CR}}$ & Direct \\
Range soundness & LB + Multi & $(m{+}2) \cdot \epsilon_{\mathrm{CR}}$ & Factor $m{+}2$ \\
\bottomrule
\end{tabular}
\end{table}

All \mhot{} proofs achieve computational soundness under the collision resistance assumption for the underlying hash function.
The reductions are tight or near-tight (with polynomial loss bounded by the number of entries), ensuring that concrete security level matches that of the hash function.

\begin{theorem}[Main Security Theorem]
\label{thm:main-security}
Let $H$ be a collision-resistant hash function with advantage bound $\epsilon_{\mathrm{CR}}$ against $\tau$-time adversaries.
Then the \mhot{} proof system is sound for all statement types against $\tau'$-time adversaries, where:
\begin{itemize}[nosep]
    \item $\tau' \approx \tau - O(\mathrm{poly}(\lambda))$ (polynomial overhead for proof verification and reduction)
    \item Soundness advantage $\leq (m + 2) \cdot \epsilon_{\mathrm{CR}}$ for proofs involving $m$ entries
\end{itemize}
In particular, for single-point proofs the reduction is tight ($\epsilon' = \epsilon_{\mathrm{CR}}$).
\end{theorem}

\begin{proof}
Follows directly from combining Lemma~\ref{lem:binding} and Theorems~\ref{thm:single-soundness-formal}, \ref{thm:multi-soundness}, \ref{thm:lb-soundness}, and \ref{thm:range-soundness-formal}.
\end{proof}


\end{document}